\documentclass[a4paper,reqno]{amsart}
\usepackage[T1]{fontenc}
\usepackage[english]{babel}
\usepackage{amssymb,bbm,enumerate}
\usepackage{xcolor}
\usepackage[linktocpage=true,colorlinks=true, linkcolor=blue, citecolor=red, urlcolor=green]{hyperref}
\usepackage{multirow} 
\usepackage{graphicx}

\usepackage{float}
\restylefloat{table}
\usepackage{tabu}

\renewcommand{\phi}{\varphi}
\renewcommand{\ker}{\Ker}
\renewcommand{\Re}{\operatorname{Re}}
\renewcommand{\Im}{\operatorname{Im}}

\newcommand{\bb}[1]{\mathbb{#1}}
\newcommand{\mc}[1]{\mathcal{#1}}
\newcommand{\mf}[1]{\mathfrak{#1}}
\newcommand{\mb}[1]{\mathbb{#1}}
\newcommand{\mbbm}[1]{\mathbbm{#1}}

\newcommand{\beq}{\begin{equation}}
\newcommand{\eeq}{\end{equation}}
\newcommand{\e}{\varepsilon}
\newcommand{\wg}{\widetilde{\mathfrak g}}
\newcommand{\wgr}{\widetilde{\mathfrak g}^{(r)}}
\newcommand{\wgz}{\widetilde{\mathfrak g}_0}
\newcommand{\wh}{\widetilde{\mathfrak h}}
\newcommand{\whz}{\widetilde{\mathfrak h}_0}
\newcommand{\gu}{\mathfrak g^{(1)}}
\newcommand{\wad}{A_{2n-1}^{(2)}}
\newcommand{\wdd}{D_{n+1}^{(2)}}
\newcommand{\lgu}{^{L}\mathfrak g^{(1)}}

\DeclareMathOperator{\at}{at}
\DeclareMathOperator{\Mat}{Mat}

\DeclareMathOperator{\End}{End}
\DeclareMathOperator{\diag}{diag}
\DeclareMathOperator{\tr}{Tr}

\DeclareMathOperator{\ad}{ad}

\DeclareMathOperator{\Ker}{Ker}

\DeclareMathOperator{\rank}{rank}
\DeclareMathOperator{\Ht}{ht}

\theoremstyle{plain}
\newtheorem{theorem}{Theorem}[section]
\newtheorem{lemma}[theorem]{Lemma}

\newtheorem{proposition}[theorem]{Proposition}

\theoremstyle{definition}
\newtheorem{definition}[theorem]{Definition}
\newtheorem{example}[theorem]{Example}

\theoremstyle{remark}
\newtheorem{remark}[theorem]{Remark}

\numberwithin{equation}{section}

\definecolor{light}{gray}{.9}

\usepackage{float}
\restylefloat{table}
\usepackage{mathtools}
\usepackage{tikz}
\usetikzlibrary{chains}

\tikzset{node distance=2em, ch/.style={circle,draw,on chain,inner sep=2pt},chj/.style={ch,join},every path/.style={shorten >=4pt,shorten <=4pt},line width=1pt,baseline=-1ex}

\let\dlabel=\alabel

\newcommand{\dnode}[2][chj]{%
\node[#1,label={below:\dlabel{#2}}] {};
}

\newcommand{\dnodebr}[1]{%
\node[chj,label={below right:\dlabel{#1}}] {};
}

\newcommand{\dydots}{%
\node[chj,draw=none,inner sep=1pt] {\dots};
}

\usepackage{datenumber}

\newcounter{dateone}
\newcounter{datetwo}

\newcommand{\difftoday}[3]{%
\setmydatenumber{dateone}{\the\year}{\the\month}{\the\day}%
\setmydatenumber{datetwo}{#1}{#2}{#3}%
\addtocounter{datetwo}{-\thedateone}%
\textcolor{red}{\the\numexpr-\thedatetwo day(s) late}
} 

\title{Bethe Ansatz and the Spectral Theory of
affine Lie algebra--valued connections II. \\ The non simply--laced case
}

\author{Davide Masoero, Andrea Raimondo, Daniele Valeri}

\address{Grupo de F\'isica Matem\'atica da Universidade de Lisboa,
Av. Prof. Gama Pinto 2, 1649-003 Lisboa, Portugal.}
\email{dmasoero@gmail.com}

\address{Dipartimento di Matematica e Applicazioni, Universit\`a degli Studi di Milano-Bicocca, Via Cozzi 53, 20125 Milano, Italy, and Grupo de F\'isica Matem\'atica da Universidade de Lisboa, Av. Prof. Gama Pinto 2, 1649-003 Lisboa, Portugal.}
\email{andrea.raimondo@unimib.it, araimondo@fc.ul.pt}

\address{Yau Mathematical Sciences Center, Tsinghua University, 100084 Beijing, China.}
\email{daniele@math.tsinghua.edu.cn}

\begin{document}

\pagestyle{plain}

\begin{abstract}
We assess the ODE/IM correspondence for the quantum $\mathfrak{g}$-KdV model, for a non-simply laced Lie algebra $\mathfrak{g}$.
This is done by studying a meromorphic connection with values in the Langlands dual algebra of the affine Lie algebra ${\mathfrak{g}}^{(1)}$, and
 constructing the relevant $\Psi$-system among subdominant solutions. We then
use the $\Psi$-system to prove that the generalized spectral determinants satisfy the
Bethe Ansatz equations of the quantum $\mathfrak{g}$-KdV model.
We also consider generalized Airy functions for twisted Kac--Moody algebras and we construct new explicit solutions to the Bethe Ansatz equations.
The paper is a continuation of our previous work on the ODE/IM correspondence for simply-laced Lie algebras.
\end{abstract}

\maketitle

\tableofcontents

\section{Introduction}\label{sec:intro}

The ODE/IM correspondence is a rich and surprising link between the theory of quantum solvable integrable models  and the spectral analysis
of linear differential operators. The origin of the correspondence goes back to \cite{doreytateo98,bazhanov01}, where it was proved that the
spectrum of certain Schroedinger operators is encoded in the Bethe Ansatz equations of the quantum KdV model. Such a discovery has boosted a remarkable
research activity,
especially in the physical literature, that did not result in a general theory but produced large number of generalizations -- see e.g.
\cite{dorey00,BLZ04,dorey07,gaiotto09,FF11,dtba,junji00,Sun12,lukyanov10,dorey13,dunning14,bazhanov14} -- linked with a variety of deep theories,
such as generalized quantum KdV, Drinfeld-Sokolov hierarchies \cite{DS85},
the geometric Langlands duality and N=4 SYM. 

The crucial step in our construction of the ODE/IM correspondence is based on the idea of Feigin and Frenkel \cite{FF11} that the differential operator corresponding to quantum $\mf g$-KdV should be an \textit{affine oper} with values in the Langlands dual of the affine Lie algebra $\gu $. Following this idea as well as some further developments contained in \cite{Sun12}, we proved in \cite{marava15} the ODE/IM correspondence in the case when $\mf g$ is simply laced. In the present paper we study the  ODE/IM correspondence for quantum $\mf g$-KdV models for a non simply-laced simple Lie algebra $\mf g$.

The integrability of the quantum $\mf g$-KdV model, first constructed in \cite{feigin96},  is expected
\footnote{The Bethe Ansatz equation (\ref{eq:TBAintronotsimply}) are expected to hold for the eigenvalues $Q^{(i)}(E)$ of
the $Q$ operators of the quantum $\mf g$-KdV model. However, these operators were constructed only in the case $\mf g=\mf{sl}_2$, in \cite{bazhanov97}, and
$\mf g= \mf{sl}_3$, in  \cite{bazhanov02integrable}.}
to be encoded in $n=\rank \mf g$ entire functions $Q^{(1)},\dots,Q^{(n)}$ 
 satisfying the $\mf g$-Bethe Ansatz equations \cite{reshetikhin87,dorey07}:
\begin{equation}\label{eq:TBAintronotsimply}
\prod_{j=1}^n\Omega^{\overline{C}_{ij}\beta_j}\frac{Q^{(j)}(\Omega^{\frac{\overline{C}_{ij}}{2}}E^\ast)}{Q^{(j)}
(\Omega^{-\frac{\overline{C}_{ij}}{2}}E^\ast)}=-1
\,,
\end{equation}
for every $E^*\in\mb C$ such that $Q^{(i)}(E^*)=0$. In equation \eqref{eq:TBAintronotsimply}, $\Omega$ and the $\beta_j$, $j=1,\dots,n$, are free parameters
and $\overline{C}=(\overline{C}_{ij})_{i,j=1}^n$ is the symmetrized Cartan matrix of the
Lie algebra $\mf g$. Note that in the simply-laced case $\overline{C}=C$ and thus (\ref{eq:TBAintronotsimply}) generalizes the simpler
ADE Bethe Ansatz \cite{zamo91,dorey07,marava15}.\\

In order to establish an ODE/IM correspondence for the Lie algebra $\mf g$,  the object to study is a meromorphic connection on the complex plane, 
introduced in \cite{FF11} (see also \cite{Sun12}),  with values in the Langlands dual Lie algebra ${}^L \gu$ of the untwisted affine Kac--Moody algebra $\gu$.
The algebra ${}^L\mf g^{(1)}$ is an affine Kac-Moody algebra of type $\wgr$,
where $\wg$ is simply-laced simple Lie algebra and $r=1,2,3$ is the order of a diagram automorphism of $\wg$, see Table \ref{table:langlands}.  The algebra ${}^L\mf g^{(1)}$ is an affine Kac-Moody algebra of type $\wgr$,
where $\wg$ is simply-laced simple Lie algebra and $r=1,2,3$ is the order of a diagram automorphism of $\wg$, see Table \ref{table:langlands}. 
The connection reads
\begin{equation}\label{eq:Lintrosl}
\mc L(x,E)=\partial_x+\frac{\ell}{x}+e+p(x,E) e_0
\,,
\end{equation}
where $\ell$ is a generic element of the Cartan subalgebra $\wh_0$ of the simple Lie algebra $\widetilde{\mf g}_0\subset \wgr$
(see Table \ref{table:graphs}), and
$e=\sum_{i=1}^{n}e_i$, where $e_0,e_1,\dots,e_{n}$
are the positive Chevalley generators of ${}^L\mf g^{(1)}$. Finally, the potential $p$ has the form
 $p(x,E)=x^{Mh^\vee}-E$, where $M>0$ and $h^\vee$ is the dual Coxeter number of $\gu$.
The choice of a potential of this form is expected to correspond to the ground state of
quantum $\mf g$-KdV \cite{dorey07,FF11}, and we stick to these potentials for simplicity. Notice, however, that all proofs work with minor modifications for the more general potentials discussed in \cite{BLZ04,FF11,bazhanov14}, which should  correspond to higher states of the theory.
As shown in Table \ref{table:langlands}, when $\mf g$ is simply-laced then the algebra
$\mf g^{(1)}$ is self-dual: ${}^L\mf g^{(1)}=\mf g^{(1)}$. Therefore, in this case we have $r=1$ and $\wg\simeq \mf g$, and the connection
(\ref{eq:Lintrosl}) coincides with the one we studied in \cite{marava15}.
Otherwise, $r>1$ and  ${}^L\mf g^{(1)}$ is a twisted affine Kac-Moody algebra.

\begin{table}[H]
\caption{Langlands correspondence for untwisted affine Lie algebras}\label{table:langlands}
{\tabulinesep=1.2mm
\begin{tabu}{ |c|c|c| }
\hline
 $\gu$   &    $ {}^L\gu =\wgr $  &  $h^\vee$  \\
\hline
\hline
$ A_n^{(1)}$, $n\geq1$  &   $A_{n}^{(1)}$ &   $n+1$  \\
\hline
$ B_n^{(1)}$, $n\geq3$  &   $A_{2n-1}^{(2)}$ &   $2n-1$  \\
\hline
$ C_n^{(1)}$, $n\geq2$ &  $D_{n+1}^{(2)}$ & $n+1$ \\
\hline
$ D_n^{(1)}$, $n\geq4$  &   $D_{n}^{(1)}$ &   $2n-2$  \\
\hline
$ E_n^{(1)}$, $n=6,7,8$  &   $ E_n^{(1)}$ &   $3(n^2-11n+34)$  \\
\hline
$ F_4^{(1)}$ & $E_{6}^{(2)}$ & $9$ \\
\hline
$ G_2^{(1)}$  &  $D_{4}^{(3)} $ & $4$\\
\hline
\end{tabu}
}
\end{table}

The $\mf g$--Bethe Ansatz will be obtained by choosing suitable finite dimensional representations of $\lgu$.
Indeed, for every finite dimensional representation of ${}^L\mf g^{(1)}$, the connection \eqref{eq:Lintrosl} yields the linear differential equation
\begin{equation}\label{eq:ODEintro}
\mc L(x,E)\Psi(x,E)=0
\,,
\end{equation}
which has two singular points: a regular singularity in  $x=0$ and an irregular singularity in $x=\infty$.
The existence of a subdominant (at $\infty$) solution to equation \eqref{eq:ODEintro} depends on the choice of the representation,
but if a  subdominant solution exists, then a natural generalization of the spectral problem for the
Schr\"odinger operator arises by considering the behavior at $0$ of the subdominant solution. Since the $\mf g$--Bethe Ansatz
will be obtained from the study of the generalized spectral problems, a first criterion to select the correct representations of
$\lgu$ is to require the equation \eqref{eq:ODEintro} to admit a subdominant solution in that representation.\\

A second crucial role in the construction of the Bethe Ansatz
is played by the so-called $\Psi-$system \cite{dorey07,marava15}.
This is a system of quadratic relations among the $n=\rank\mf g$ subdominant solutions $\Psi^{(i)}(x,E)$ defined on distinguished 
representations $V^{(i)}$, $i=1,\dots,n$, of $\lgu$.
It reads
\begin{equation}\label{eq:PsiIntro}
m_i\left(R_i\big(\psi^{(i)}_{-\frac{D_i}{2}}\big)\wedge\psi^{(i)}_{\frac{D_i}{2}}\right)=
\bigotimes_{j\in I}\bigotimes_{\ell=0}^{B_{ij}-1}\psi^{(j)}_{\frac{B_{ij}-1-2\ell}{2r}}
\,, \qquad i=1,\dots,n,
\end{equation}
where $D_i$ is the $i$-th element of the symmetrizing matrix
$D=\mbox{diag }(D_1,\dots,D_n)$, namely
$\overline{C}=DC$, where $C$ is the Cartan matrix of $\mf g$,
and $B=2\mbbm1_n-C$ denotes the incidence matrix of $\mf g$.
In addition, $R_i$ is a certain isomorphism of representations of ${}^L\mf g^{(1)}$, which reduces to the identity map if $D_i=1$.
In the simply-laced case, the $\Psi-$system \eqref{eq:PsiIntro} coincides with the one studied in our previous paper \cite{marava15}.
The existence of the $\Psi-$system was conjectured in \cite{dorey07} (and there proved in the case $A_n$) and proved for a
simply-laced Lie algebra in
\cite{marava15}.
The construction of the $\Psi-$system for non simply-laced Lie algebras was unknown before the present paper.\\

Although the problem considered in the present work shares some similarities with the simply-laced case, we stress that
the extension of the ODE/IM correspondence to the case of a non simply-laced Lie algebra $\mf g$ is  highly non trivial,
as the appearance of the twisted Kac-Moody algebra ${}^L\mf g^{(1)}=\widetilde{\mf g}^{(r)}$  provides new difficulties.
First, in the simply-laced
case we chose the representations $V^{(i)}$ of the untwisted Kac-Moody algebra ${}^L\mf g^{(1)}=\mf g^{(1)}$, $i=1,\dots, n=\rank\mf g$, to be some evaluation
representations of the fundamental representations of $\mf g$.
To construct the
$\Psi$-system and the Bethe Ansatz for a non simply-laced Lie algebra $\mf g$,
we need to properly select the representations $V^{(i)}$ of the twisted Kac-Moody algebra ${}^L\mf g^{(1)}=\wgr$, for $i=1,\dots,n=\rank\mf g$,
among the evaluation representations of the $\rank\wg >n$ fundamental representations of $\widetilde{\mf g}$.

Another problem arises in the study of the element $\Lambda=\sum_{i=0}^n e_i\in{}^L\mf g^{(1)}$,
whose spectrum in the representations $V^{(i)}$'s
plays an extremely important role in the proof of the ODE/IM correspondence.
In our previous work \cite{marava15}, we were able to explicitly diagonalize $\Lambda$
due to the fact that in the case of an untwisted Kac-Moody algebra $\Lambda$ is the eigenvector of a Killing-Coxeter element.
However, for twisted Kac-Moody algebras,
$\Lambda$ is not the eigenvector of a Killing-Coxeter element
(nor a twisted Killing-Coxeter element), and the study of its spectrum requires a different approach.
Finally, the Langlands duality, which was hidden in the simply-laced case, has now to be explicitly taken in account, and the
whole construction of the correspondence from the connection \eqref{eq:Lintrosl} to the Bethe Ansatz \eqref{eq:TBAintronotsimply} has to be modified accordingly.

\medskip

For the sake of completeness, and to provide a unifying presentation to the ODE/IM correspondence, the results of the present paper are stated for an arbitrary simple Lie algebra. However, the proofs are provided in detail for the non simply laced case only,  as in the simply laced case they were already obtained in \cite{marava15}.\\

The paper is organized as follows. In Section \ref{sec:liealgebra} we review some
basic facts about the theory of Lie algebras, diagram automorphisms of Lie algebras, Kac-Moody Lie algebras and 
their finite dimensional representations.

In Section \ref{sec:L} we review the asymptotic analysis of equation \eqref{eq:ODEintro}, following \cite{marava15}.
The main result is provided by Theorem \ref{thm:asymptotic} about the existence of the subdominant solution
of the equation (\ref{eq:ODEintro}). It states that if in a finite-dimensional representation the element $\Lambda$
has a maximal eigenvalue (see Definition \ref{def:maximal}),
then equation (\ref{eq:ODEintro}) admits a unique subdominant
solution $\Psi(x,E)$.

Section \ref{sec:psi} is devoted to the construction of the $\Psi$-system \eqref{eq:PsiIntro}.
For every node $i=1,\dots,n$ of the Dynkin diagram of $\mf g$,
we define a distinguished finite dimensional representation $V^{(i)}$ of ${}^L\mf g^{(1)}$.
The construction of these representations relies on the definition of a good vertex of a Dynkin diagram with respect
to a diagram automorphism, which we introduce in Definition \ref{def:good}.
In Theorem \ref{thm:genlamba} we claim
that in each representation $V^{(i)}$ the element $\Lambda$ has a maximal eigenvalue
$\lambda^{(i)}$, where $\lambda^{(1)}=1$, and that the following
remarkable identity holds:
\begin{equation}\label{eq:genlambdaintro}
\left(e^{-\frac{\pi\sqrt{-1} D_i}{h^\vee}}+e^{\frac{\pi\sqrt{-1} D_i}{h^\vee}}\right)\lambda^{(i)}
=\sum_{j=1}^n\left(\sum_{\ell=0}^{B_{ij}-1}e^{\frac{\pi\sqrt{-1} (B_{ij}-1-2\ell)}{r h^\vee}}\right)\lambda^{(j)}
\,.
\end{equation}
As a consequence of these results we get the existence of the fundamental solutions $\Psi^{(i)}(x,E)$, for every
$i=1,\dots,n$, as well as of the
$\Psi$-system \eqref{eq:PsiIntro}, which is provided by Theorem \ref{thm:psi-sistem}.

We prove Theorem \ref{thm:genlamba} for non simply-laced Lie algebras  by a case-by-case inspection in Section
\ref{sec:proof_main}.
In the simply laced-case a proof was given in \cite[Proposition 3.4]{marava15}.

In Section \ref{sec:qsystem} we derive the Bethe Ansatz equations \eqref{eq:TBAintronotsimply}.
To this aim, as done in \cite{marava15}, we study the local behavior of equations \eqref{eq:ODEintro} close to the Fuchsian
singularity $x=0$, and we define the generalized spectral determinants $Q^{(i)}(E;\ell)$ and $\widetilde{Q}^{(i)}(E;\ell)$.
Using the $\Psi$-system \eqref{eq:PsiIntro} we prove Theorem \ref{thm:QQtilde}, which gives a set of
quadratic relations among the spectral determinants, known as $Q\widetilde{Q}$-system.
%
In Theorem \ref{thm:betheansatz}, evaluating the $Q\widetilde{Q}$-systems at the zeros of the functions $Q^{(i)}(E;\ell)$, we obtain
the Bethe Ansatz equations \eqref{eq:TBAintronotsimply}. We also investigate the action of the Weyl group of $\widetilde{\mf g}_0$ on the space of solutions to the Bethe Ansatz, and this provides a set of new solutions.

Finally, in Section \ref{app:airy}, we study an integral representation of the subdominant
solution of equation \eqref{eq:ODEintro}, in the case of a non simply-laced Lie algebra $\mf g$,
with a linear potential $p(x,E)=x$ and $\ell=0$.
Although this case is not generic, we can anyway define the spectral determinants $Q^{(i)}(E)=Q^{(i)}(E;0)$ and
we provide their expression in terms of Airy functions associated to the twisted Kac-Moody algebra ${}^L\mf g^{(1)}$.
In this way, we provide new examples of exact solution to the Bethe Ansatz equations, others than the ones already known in the literature \cite{junji15}.

\subsection*{Parameters of the ODE/IM correspondence}

We conclude the introduction by showing the exact relation among the parameters $\Omega$,  $\beta_1,\dots,\beta_n$, appearing in the  Bethe Ansatz equation (\ref{eq:TBAintronotsimply})
and the parameters $M$ and $\ell \in \whz$ of the connection (\ref{eq:Lintrosl}).
The first relation is very simple, and reads
\begin{equation}\label{eq:OmegaIntro}
 \Omega=e^{i\frac{2\pi M}{M+1}} \; .
\end{equation}
The relation between the phases $\beta_j$'s and the parameters of the connection is more involved. On the integrable systems side,
the phases $\beta_j$'s parametrize
the possible twisting of periodic boundary conditions, 
and from \cite{reshetikhin87} we know that the latter can be described in terms of elements of the Cartan subalgebra $\mf h$ of $\mf g$.
It then follows that the phases $\beta_j$'s naturally belong to $\mf{h}^* \cong \bb{C}^n$,
where $n=\rank \mf g$. Note that the element $\ell$ belongs to the Cartan subalgebra $\whz \subset \wgz \subset \wgr$,
which also has dimension $n$.
If we choose $\ell \in \whz$ to be generic, namely to belong to the open convex dual of a Weyl chamber $\mf{W}_\ell$, then such $\ell$ is
associated with the element $w_\ell$ of the Weyl group of $\wgz$ that maps the principal
Weyl chamber to $\mf{W}_\ell$. For any $\ell$ in generic position we have
\begin{equation}\label{eq:betajIntro}
 \beta_j  =\frac{1}{2Mh^\vee}w_\ell(\omega_j)(\ell+h)
\,,
\end{equation}
where $\omega_j$ is the $j\mbox{-th}$ fundamental weight of $\whz$ and $h$ is the unique
element of $\whz$ satisfying the commutation relations
$[h,e_i]=e_i$, for every $i=1,\dots,n$. The relation between $\beta_j$'s and $\ell$ is only piecewise linear, due to the action of the Weyl group of $\wgz$.
However, the Weyl group of $\wgz$ is isomorphic to the Weyl group of $\mf g$, and the latter naturally acts on the space of solutions
of the Bethe Ansatz equations \eqref{eq:TBAintronotsimply}, see Subsection \ref{sub:weyl}. 


\subsection*{Acknowledgments}
The authors were partially supported by the INdAM-GNFM ``Progetto Giovani 2014''.
D. M. is supported by the
FCT scholarship, number SFRH/BPD/75908/2011.  A.R. is supported by the
FCT project \lq\lq Incentivo/MAT/ UI0208/2014\rq\rq. D.V. is supported by an NSFC \lq\lq Research Fund for International Young Scientists\rq\rq\, grant. We thank Edward Frenkel for useful discussions. A.R would like to thank the Yau Mathematical Sciences Center and the Dipartimento di Matematica of the University of Genova for the kind hospitality.
D. M. would like to thank Y. Shi for useful discussions about the action of the Weyl group on tuples of weights.

\section{Twisted Kac-Moody algebras and finite dimensional representations}\label{sec:liealgebra}
As stated in the Introduction, the ODE/IM correspondence for a simple Lie algebra $\mf g$ requires the study of a connection with values in the Langlands dual Lie algebra ${}^L\mf g^{(1)}$ of the untwisted Kac-Moody algebra $\gu$.
As shown in Table \ref{table:langlands}, the  algebra ${}^L\mf g^{(1)}$ is a Kac-Moody algebra of type $\widetilde{\mf g}^{(r)}$, where
$\widetilde{\mf g}$ is a simply-laced Lie algebra and $r$ is the order of a diagram automorphism $\sigma$. To any simple Lie algebra $\mf g$
we thus associate a unique pair $(\wg, r)$, as follows:
\begin{equation}\label{20151022:eq1}
\mf{g}\; \longrightarrow \;\mf{g}^{(1)}\; \longrightarrow \;^{L}\mf{g}^{(1)}=\wgr\; \longrightarrow\; (\wg, r).
\end{equation}
For $\mf g$ simply-laced, which is the case considered in \cite{marava15},  we have ${}^L\mf g^{(1)}=\mf g^{(1)}$, so that the diagram
automoprhism has order $r=1$ (the identity),  and the above correspondence becomes trivial: $\mf{g} \longrightarrow (\mf{g}, 1)$.
For $\mf g$ non simply-laced we have $r>2$, and the list of non-trivial diagram automorphisms of $\widetilde{\mf g}$ obtained by the
correspondence \eqref{20151022:eq1} is provided in Table \ref{table:graphs}.

In this section we review some facts about diagram automorphisms of simple Lie algebras
and we describe how the Cartan matrix and the Dynkin diagram of  $\mf g$ can be recovered by means of $\sigma$ from the Cartan
matrix and the Dynkin diagram of $\wg$.
Moreover, we introduce those aspects of the representation theory of ${}^L\mf g^{(1)}$ which we will use
throughout the paper.

\subsection{Simple Lie algebras and Dynkin diagram automorphisms}\label{sec:0}
Let $\widetilde{\mf g}$ be a simply-laced Lie algebra, with Dynkin diagram as given  in Table \ref{fig:dynkin}.
Set
$\widetilde{I}=\left\{1,\dots,\widetilde{n}=\rank\wg\right\}$, let $\widetilde{C}=(\widetilde C_{ij})_{i,j\in\widetilde I}$ be the Cartan matrix of $\wg$ 
and $\widetilde B=2\mbbm1_{\widetilde n}-\widetilde C$  be the incidence matrix of the Dynkin diagram of $\widetilde{\mf g}$.
We denote by $\{\widetilde e_i,\widetilde h_i,\widetilde f_i\mid  i\in \widetilde{I}\}\subset\widetilde{\mf g}$ the
set of Chevalley generators of $\widetilde{\mf g}$. They satisfy the relations ($i,j\in \widetilde{I}$)
\begin{equation}\label{eq:chevalley}
[\widetilde h_i,\widetilde h_j]=0\,,
\quad
[\widetilde h_i,\widetilde e_j]=\widetilde{C}_{ij}\widetilde e_j\,,
\quad
[\widetilde h_i,\widetilde f_j]=-\widetilde{C}_{ij}\widetilde f_j\,,
\quad
[\widetilde e_i,\widetilde f_j]=\delta_{ij} \widetilde h_i\,.
\end{equation}
Recall that a diagram automorphism $\sigma$ is a permutation on the set $\widetilde I$ such that
$\widetilde C_{\sigma(i),\sigma(j)}=\widetilde C_{ij}$.
It is well known that a diagram automorphism $\sigma$ extends to a Lie algebra automorphism (which we still denote by $\sigma$)
$\sigma:\widetilde{\mf g}\rightarrow\widetilde{\mf g}$ defined on the Chevalley generators
by ($i\in\widetilde{I}$)
$$
\sigma(\widetilde e_i)=\widetilde e_{\sigma(i)}\,,
\qquad
\sigma(\widetilde h_i)=\widetilde h_{\sigma(i)}\,,
\qquad
\sigma(\widetilde f_i)=\widetilde f_{\sigma(i)}
\,.$$
The diagram automorphism $\sigma$, of order $r$, induces on $\wg$ the following gradation
\begin{equation}\label{eq:201506251}
\widetilde{\mf{g}}=\bigoplus_{k\in \bb{Z}/r\mb Z} \widetilde{\mf{g}}_k
\,,
\qquad\text{where}\qquad\widetilde{\mf{g}}_k
=\left\{x\in \widetilde{\mf g}\, |\, \sigma(x)=e^{\frac{2\pi i k}{r}}\,x\right\}.
\end{equation}
It is well known that $\wgz$ -- the invariant subalgebra under the action of $\sigma$ -- is a simple Lie algebra, see Table \ref{table:graphs}.
\begin{table}[H]
\caption{Non-trivial Dynkin diagram automorphisms for simple Lie algebras.}\label{table:graphs}
{\tabulinesep=0.6mm 
\begin{tabu}{ |c|c|c|c|c|c| } 
\hline 
$\mf g$ & $\widetilde{\mf{g}}$ & $\wg_0$ & $\sigma$&   $r$ &   $D_1,\dots, D_n$\\ 
\hline 
\hline
&&&&&\\
$B_n$ & $A_{2n-1}$ & $C_n$ & $\sigma(i)=2n-i$ & $2$ &   $1,\,\dots, \,1,\, \frac{1}{2}$\\
&&&&&\\
\hline
\multirow{3}{*}{ $C_n$}
& & &$\sigma(i)=i, \;\; 1\leq i\leq n-1$ & &\\
 & $D_{n+1}$ & $B_n$ & $\sigma(n)=n+1$ & $2$ &  $\frac{1}{2},\,\dots,\,\frac{1}{2},\,1$ \\
& & & $\sigma(n+1)=n$ & &\\
\hline
\multirow{3}{*}{ $F_4$ }
&&& $\sigma(1)=6$ \quad $\sigma(6)=1$ & &\\
& $E_{6}$ & $F_4$ & $\sigma(2)=5$\;\; \;\,$\sigma(5)=2$ & $2$ &  $1,\,1,\,\frac{1}{2},\,\frac{1}{2}$\\
&&& $\sigma(3)=3$ \quad $\sigma(4)=4$ & &\\
\hline
\multirow{2}{*}{ $G_2$}
& \multirow{2}{*}{$D_4$}& \multirow{2}{*}{$G_2$} & $\sigma(1)=3$ \quad $\sigma(3)=4$ &\multirow{2}{*}{$3$} & \multirow{2}{*}{$1,\,\frac{1}{3}$} \\
&&& $\sigma(4)=1$\;\;\;\, $\sigma(2)=2$ &  &\\
\hline
\end{tabu}
}
\end{table}
Let $\widetilde I^\sigma$ denote the set of orbits in $\widetilde{I}$ under the action of the permutation $\sigma$,
and let
$
I=\{\min_{j\in J} j\}_{J\in \widetilde I^\sigma}\subset \widetilde{I}
\,.
$
For every $i\in\widetilde{I}$ we also denote by $\langle i \rangle\in \mb Z_+$
the cardinality of the orbit under the action of $\sigma$ containing $i$
and we set
\begin{equation}\label{Di}
D_i=\frac{\langle i\rangle}{r}
\,.
\end{equation}
The following elementary result is crucial for our purposes.
\begin{proposition}\label{prop:gfromgtilde}
Let $\mf g$ be a simple Lie algebra with Cartan matix $C$, and let $B$ be the incidence matrix of its Dynkin diagram.
Let $\widetilde{\mf g}$ be the simply-laced Lie algebra  and $\sigma$ the diagram automorphism of order $r$ corresponding to $\mf g$ through the map \eqref{20151022:eq1}. The following facts hold:
\begin{enumerate}[i)]
\item
$I=\lbrace 1,\dots, n \rbrace$, where $n= \rank \mf g$. 
\item
The Cartan matrix $C$ can be obtained summing over the rows of $\widetilde{C}$ along the orbits of $\sigma$. Namely,
\begin{equation}\label{Crowsum}
C_{ij}=\sum_{\ell=0}^{\langle j \rangle -1}\widetilde{C}_{i\sigma^\ell(j)}, \, \mbox{ for all } i,j \in I
\,.
\end{equation}
\item
The incidence matrix $B$ can be obtained summing over the rows of $\widetilde B$ along the orbits of $\sigma$.
Namely,
\begin{equation}\label{Browsum}
B_{ij}=\sum_{\ell=0}^{\langle j \rangle -1}\widetilde{B}_{i\sigma^\ell(j)}, \, \mbox{ for all } i,j \in I
\,.
\end{equation}
\item Let $D=\diag\left(D_1,\dots, D_n\right)$, where $D_i$, $i\in I$, is defined by equation \eqref{Di}.
Then the matrix $\overline{C}=DC$ is symmetric.
\end{enumerate}
\end{proposition}
\begin{proof} If $\mf g$ is simply-laced then $r=1$, $\sigma$ is the identity automorphism and there is nothing to prove. In particular, in this case we have $I=\widetilde{I}$, $C=\widetilde{C}$, $B=\widetilde{B}$ and $D_i=1$ for all $i\in I$. I f $\mf{g}$ is non simply-laced, then parts i), ii) and iii) have been proved in \cite{fuscsc96}, and part iv) can be checked by a direct computation.
\end{proof}
\begin{remark}
In \cite{fuscsc96}, the simple Lie algebra $\mf g$ whose Cartan matrix is obtained by the Cartan matrix of $\wg$ as in \eqref{Crowsum} was called the \emph{orbit Lie algebra of $\wg$ with $\sigma$}. As already noted in \cite{fuscsc96}, $\mf g$ is not constructed as a subalgebra of $\wg$, and it does not need to be isomorphic to the fixed point subalgebra $\wg_0$, see Table \ref{table:graphs}.
\end{remark}
\begin{table}[H]
\caption{Dynkin diagrams of simple Lie algebras of $ADE$ type.}
\label{fig:dynkin}
\begin{align*}
&
A_{\widetilde{n}}
\quad
\begin{tikzpicture}[start chain]
\dnode{1}
\dnode{2}
\dydots
\dnode{\widetilde{n}-1}
\dnode{\widetilde{n}}
\end{tikzpicture}
&\quad\quad&
E_6
\quad
\begin{tikzpicture}
\begin{scope}[start chain]
\dnode{1}
\dnode{2}
\dnode{3}
\dnode{5}
\dnode{6}
\end{scope}
\begin{scope}[start chain=br going above]
\chainin (chain-3);
\dnodebr{4}
\end{scope}
\end{tikzpicture}
\\
&
D_{\widetilde{n}}
\quad
\begin{tikzpicture}
\begin{scope}[start chain]
\dnode{1}
\dnode{2}
\node[chj,draw=none] {\dots};
\dnode{\widetilde{n}-2}
\dnode{\widetilde{n}-1}
\end{scope}
\begin{scope}[start chain=br going above]
\chainin(chain-4);
\dnodebr{\widetilde{n}}
\end{scope}
\end{tikzpicture}
&\quad\quad&
E_7
\quad
\begin{tikzpicture}
\begin{scope}[start chain]
\foreach \dyni in {1,2,3,4,6,7} {
\dnode{\dyni}
}
\end{scope}
\begin{scope}[start chain=br going above]
\chainin (chain-4);
\dnodebr{5}
\end{scope}
\end{tikzpicture}
\\
&
&\quad\quad&
\!\!\!\!\!\!\!\!\!\!\!\!\!\!\!\!\!\!\!\!\!\!\!\!\!\!\!\!\!\!\!\!
\!\!\!\!\!\!\!\!\!\!\!\!\!\!\!\!\!\!\!\!\!\!\!\!\!\!\!\!\!\!\!\!
E_8
\quad
\begin{tikzpicture}
\begin{scope}[start chain]
\foreach \dyni in {1,2,3,4,5,7,8} {
\dnode{\dyni}
}
\end{scope}
\begin{scope}[start chain=br going above]
\chainin (chain-5);
\dnodebr{6}
\end{scope}
\end{tikzpicture}
&
\end{align*}
\end{table}

\subsection{Basic facts about representation theory of Lie algebras}\label{sec:repwg}

As in Section \ref{sec:0}, let $\widetilde{\mf g}$ be a simply-laced Lie algebra, and let us denote by
$\widetilde{\mf h}\subset\widetilde{\mf g}$ its Cartan subalgebra.
We denote by $R\subset\widetilde{\mf h}^*$ the set of roots of $\widetilde{\mf g}$ and by
$\Delta=\{\alpha_i\mid i\in \widetilde I\}\subset R$
the set of simple roots. Also, we denote by $P\subset\widetilde{\mf h}^*$ the set of weights of 
$\widetilde{\mf g}$ and by $P^+\subset P$ the set of dominant weights. If $\omega\in P^+$, we denote by $L(\omega)$ the irreducible highest weight representation with highest weight $\omega$, and we denote by $P_\omega\subset P$ the set of weights of $L(\omega)$.
Recall that the fundamental weights of $\wg$ are those elements $\omega_i\in P^+,\, i\in \widetilde I$,  satisfying
\begin{equation}\label{20150107:eq1}
\omega_i(\widetilde h_j)=\delta_{ij}\,,
\qquad
\text{for every }j\in \widetilde{I}
\,.
\end{equation}
The corresponding highest weight representations  $L(\omega_i)$, $i\in \widetilde{I}$,  are known as fundamental
representations of $\wg$, and for every $i\in \widetilde{I}$  we denote by $v_i\in L(\omega_i)$ the highest weight
vector of the representation $L(\omega_i)$.
Hence, we naturally associate to the $i-$th vertex of the Dynkin diagram of $\wg$
the corresponding fundamental representation $L(\omega_i)$ of $\wg$.
Let us consider the dominant weight
\begin{equation}\label{eq:etai}
\eta_i=\sum_{j\in \widetilde I}\widetilde B_{ij}\omega_j, \qquad i\in \widetilde{I} \; .
\end{equation}
Recall from \cite{marava15} that we can find a unique copy of $L(\eta_i)$, $i\in\widetilde I$, as irreducible component of the representation $\bigwedge^2L(\omega_i)$ as well as of the representation
$\bigotimes_{j\in \widetilde{I}}L(\omega_j)^{\otimes \widetilde B_{ij}}$.
We can thus decompose the representation $\bigwedge^2L(\omega_i)$ as
\begin{equation}\label{20151024:eq1}
\bigwedge^2L(\omega_i)=L(\eta_i)\oplus U\,,
\end{equation}
where $U$ is the direct sum of all the irreducible representations different from $L(\eta_i)$, and the subrepresentation
isomorphic to $L(\eta_i)$ is generated by the highest weight vector $\widetilde f_iv_i \wedge v_i$.
It follows from this that for every $i\in\widetilde I$, there exists
a unique morphism of representations of $\widetilde{\mf g}$:
\begin{equation}\label{morphism:090715}
\widetilde{m}_i=\bigwedge^2 L(\omega_i)\longrightarrow \bigotimes_{j\in\widetilde{I}}L(\omega_j)^{\otimes \widetilde{B}{ij}},
\end{equation}
such that $\ker \widetilde m_i= U$ and $\widetilde{m}_i(\widetilde f_iv_i\wedge v_i)=\otimes_{j\in\widetilde{I}}v_j$.
\medskip

We now consider the action of the diagram automorphism $\sigma$ on $\wg$-modules. Let $V$ be a $\widetilde{\mf{g}}$-module, so that we have a Lie algebra homomorphism
$\rho :\widetilde{\mf{g}}\to \End(V)$. Then,  the composition $\rho^\sigma=\rho\circ\sigma:\widetilde{\mf g}\to \End(V)$
is a Lie algebra homomorphism; we denote by $V^{\sigma}$ the vector space $V$ with the $\widetilde{\mf g}$-module structure given by $\rho^\sigma$.
Note that $V^\sigma$ is irreducible if and only if $V$ is irreducible.
The following result shows that $L(\omega_i)^\sigma$ is 
a fundamental representation of $\widetilde{\mf g}$.
\begin{lemma}\label{lem:300415} 
We have the following isomorphism of representations
$$
L(\omega_i)^\sigma\simeq L(\omega_{\sigma^{-1}(i)})
\,.
$$
\end{lemma}
\begin{proof}
By definition, the highest weight vector $v_i\in L(\omega_i)$ satisfies the conditions
$$
\rho(\widetilde e_j)v_i=0\,,
\qquad \rho(\widetilde h_j)v_i=\omega_i(\widetilde h_j)v_i=\delta_{ij}v_i\,, \qquad j\in \widetilde{I}\,.
$$  
Hence, the representation $\rho^\sigma$ acts on $v_i$ as follows ($j\in\widetilde J$):
$$
\rho^\sigma(\widetilde e_j)v_i=\rho(\sigma(\widetilde e_j))v_i=\rho(\widetilde e_{\sigma(j)})v_i=0
\,,
$$
and
$$
\rho^\sigma(\widetilde h_j)v_i=\rho(\widetilde h_{\sigma(j)})v_i
=\omega_i(\widetilde h_{\sigma(j)})v_i=\delta_{i\sigma(j)}v_i
=\delta_{\sigma^{-1}(i)j}v_i=\omega_{\sigma^{-1}(i)}(\widetilde h_j)v_i
\,.
$$
This shows that $v_i$ is a highest weight vector for $L(\omega_i)^\sigma$, of weight $\omega_{\sigma^{-1}(i)}$.
Since $L(\omega_i)$ is irreducible so is $L(\omega_i)^\sigma$, and this implies that $L(\omega_i)^\sigma$ is isomorphic to
the fundamental module $L(\omega_{\sigma^{-1}(i)})$.
\end{proof}

\subsection{Twisted affine Kac-Moody algebras and finite dimensional representations}
Let $\mf g$ be a simple Lie algebra and $(\wg,r)$ the pair associated to it by the map \eqref{20151022:eq1}. In this section we review the
loop algebra realization of the 
affine Kac-Moody algebra ${}^L\mf g^{(1)}=\widetilde{\mf g}^{(r)}$ and we define its
finite dimensional representations which will be of interest for this paper.
The presentation below
is given for the twisted case, but it reduces to the untwisted case when $r=1$. We follow mainly \cite{kac90}, to which we refer for further
details.
\medskip

Let $\mc L(\widetilde{\mf{g}})=\widetilde{\mf{g}}\otimes \bb{C}[t,t^{-1}]$ denote the loop algebra of
$\widetilde{\mf{g}}$. The Lie algebra structure of $\widetilde{\mf g}$ extends to a Lie algebra structure on
$\mc L(\widetilde{\mf g})$ in the obvious way.
We  extend $\sigma$ to a Lie algebra homomorphism
(which we still denote by $\sigma$) $\sigma:\mc L(\widetilde{\mf{g}})\to  \mc L(\widetilde{\mf g})$
 given by
\begin{equation}\label{sigmaonkm}
\sigma(x\otimes f(t))=\sigma(x)\otimes f(e^{-\frac{2\pi i}{r}}t),
\end{equation}
for $x\in \wg$, $f\in \bb{C}[t,t^{-1}]$.
The subalgebra of invariants with respect to $\sigma$ is known as (twisted) loop algebra and we denote it by
$$
\mc L(\wg,r)= \mc L(\widetilde{\mf g})^\sigma
=\left\{y\in  \mc L(\widetilde{\mf{g}})\mid \sigma(y)=y\right\}
\,.
$$
The gradation \eqref{eq:201506251} of $\wg$, together with the action \eqref{sigmaonkm}, induces on the twisted loop algebra the following gradation

$$
\mc L(\wg,r)=\bigoplus_{k\in \bb{Z}/r\mb Z} \widetilde{\mf{g}}_k\otimes t^k\bb{C}[t^r,t^{-r}]
\,.
$$
The twisted affine Kac-Moody algebra $\wgr=\mc L(\wg,r)\oplus\bb{C}c$
is obtained as the unique central extension of $\mc L(\wg,r)$ by a central element $c$.
The Chevalley generators $\{e_i,h_i,f_i\mid i=0,\dots,n\}\subset\widetilde{\mf g}^{(r)}$ can be obtained as follows. The generators $e_i$, $h_i$, and $f_i$, for $i\in I$, are obtained as
linear combinations of the Chevalley generators of $\widetilde{\mf g}$:
\begin{equation}\label{chevg0}
e_i=\sum_{\ell=0}^{\langle i\rangle -1}\widetilde{e}_{\sigma^\ell(i)}, \quad f_i=\sum_{\ell=0}^{\langle i\rangle -1}\widetilde{f}_{\sigma^\ell(i)}, \quad h_i=\sum_{\ell=0}^{\langle i\rangle -1}\widetilde{h}_{\sigma^\ell(i)}.
\end{equation}
Moreover, they generate the simple Lie algebra $\widetilde{\mf g}_0$. The generator $e_0$ (respectively $f_0$) is of the form
$$
e_0=a\otimes t\quad\text{ (respectively }f_0=a\otimes t^{-1})\,,
$$
where $a\in\widetilde{\mf g}_{1}$ (respectively $a\in\widetilde{\mf g}_{-1}$) is a lowest (respectively highest)
weight vector with respect to the action of $\widetilde{\mf g}_0$.
Finally, the generator $h_0$ is obtained as a linear combinations of the generators $h_i$, $i\in I$, and the central
element $c$.
\medskip

Let $\rho: \wg \to \End(V)$ be a finite dimensional representation of $\wg$.
For $k \in \bb{C}$ we define a finite dimensional representation $\rho_k:\wgr \to\End(V_k)$ of $\wgr$ in the following way:
as a vector space we take $V_k= V$, and the map $\rho_{k}$ is defined 
by
\begin{align*}
&\rho_{k}(a \otimes \varphi(t))v=\varphi(e^{2\pi i k})(\rho(a)v), &\text{for }a\in\wg\,,\varphi\in\mb C[t,t^{-1}]\,, v\in V,\\
&\rho_k(c)v=0, &\text{for } v\in V.
\end{align*}
The representation $V_k$ is known as an evaluation representation of $\widetilde{\mf g}^{(r)}$ at $t=e^{2\pi i k}$. 
\begin{remark} Since the central element $c$ acts trivially on evaluation representations, these are level zero representations. More precisely, evaluation representations are level zero irreducible finite dimensional representations of $\wgr$. Moreover, every level zero finite dimensional irreducible representation of $\wgr$ can be obtained as a tensor product of evaluation representations \cite{rao93}. 
\end{remark}

Note that $(\rho_k)^\sigma\neq(\rho^\sigma)_k$ (we are using the same notation
introduced in Section \ref{sec:repwg}).
Since we will be interested only in evaluation representations of $\widetilde{\mf g}$-modules,
and not in their twisting by the action of $\sigma$,
we will always denote by $V_k^\sigma$ the  evaluation representation $\rho_k^\sigma=(\rho^\sigma)_k$.
\begin{proposition}\label{prop: 120515}
Let $V$ be a finite dimensional representation of $\wg$, and let $L(\omega_i)$, $i\in \widetilde{I}$, be a
fundamental representations of $\wg$.
\begin{enumerate}[i)]
 \item For every $k\in\bb{C}$, we have $V^{\sigma}_k \simeq V_{k+\frac{1}{r}}$, where $r$ is the order of $\sigma$.
 \item For every $i\in \widetilde{I}$ and $k\in\bb{C}$,
 $L(\omega_{\sigma^{-1}(i)})_k\simeq L(\omega_i)_{k+\frac{1}{r}}$.
 \item For every $i\in I$ and $k\in\bb{C}$, there exists an isomorphism of evaluation representations
\begin{equation}\label{isoR_i}
R_i:L(\omega_{i})_k\longrightarrow L(\omega_i)_{k+D_i},
\end{equation}
where $D_i$ is defined in \eqref{Di}.
\end{enumerate}
 \end{proposition}
 \begin{proof}
A generic element of $\widetilde{\mf{g}}^{(r)}$ is of the form $x_m\otimes t^{lr+m}$,
where $x_m\in \widetilde{\mf{g}}_m$.
Then we have that ($v\in V$)
\begin{align*}
\rho^\sigma_{k}(x_m\otimes z^{lr+m})v&=e^{2\pi i k(lr+m)}\rho^\sigma(x_m)v
\\
&=e^{2\pi i (t+\frac{1}{r})(lr+m)}\rho(x_m)v
=\rho_{k+\frac{1}{r}}(x_m\otimes z^{lr+m})v
\,,
\end{align*}
where in the second identity we used the fact that $x_m\in \widetilde{\mf{g}}_m$. This proves part i).
Applying part i) to the evaluation representation $L(\omega_i)_k$ we get
$L(\omega_i)^\sigma_k \simeq L(\omega_i)_{k+\frac{1}{r}}$. Hence, part ii) follows by Lemma \ref{lem:300415}.
Part iii) follows applying $\langle i \rangle$ times part ii).
\end{proof}
\begin{remark}
Note that the isomorphism $R_i$ given in equation \eqref{isoR_i} reduces to the identity when $D_i=1$.
\end{remark}
\subsection{The cyclic element $\Lambda$} For any simple Lie algebra $\mf g$, we define the element
\begin{equation}\label{eq:defLambda}
\Lambda=\sum_{i=0}^n e_i\in{}^L\mf g^{(1)}
\,,
\end{equation}
the sum of the positive Chevalley generators of ${}^L\mf g^{(1)}=\mf{\widetilde{g}}^{(r)}$,
which will play an important role in order to derive the main results
in Sections \ref{sec:L} and \ref{sec:psi}.
Let $h^\vee$ denote the dual Coxeter number of $\mf g^{(1)}$
(equivalently, $h^\vee$ is the Coxeter number of ${}^L\mf g^{(1)}$),
 and let $h\in\wh$ be the unique element such that (see e.g. \cite{CMG93},\cite{kac90})
\begin{equation}\label{eq:hrelations}
[h,e_i]=e_i, \quad i \in \ I
\,,
\qquad
[h,e_0]=-(h^\vee-1)e_0
\,.
\end{equation}
Then, $\Lambda$ is an eigenvector with eigenvalue $e^{\frac{2\pi i}{h^\vee}}$
of the inner automorphism $e^{\frac{2\pi i}{h^\vee}\ad h}$. Indeed, introducing the
${}^L\mf g^{(1)}$-automorphism  $\mc{M}_{k}$ which
fixes $\widetilde{\mf g}$ and $c$ and sends $t \to e^{2\pi i k} t$, then we have
\begin{equation}\label{20151003:eq1}
\gamma^{k\ad h}\Lambda=\gamma^k\mc M_{-k}(\Lambda) \, , \quad \gamma=e^{\frac{2 \pi i}{h^\vee}},
\end{equation}
for any  $k\in\mb C$. Denoting by $\text{spec}\left(\Lambda, V_k\right)$ the spectrum of $\Lambda$  in an evaluation representation $V_k$, we have that
$$\text{spec}\left(\Lambda, V_{k+s}\right)=\gamma^s \,\text{spec}\left(\Lambda, V_k\right),$$
from which it follows that in any evaluation representation the spectrum of $\Lambda$ is invariant under multiplication by $\gamma$. Moreover,
using Proposition \ref{prop: 120515} (iii) we easily see that, for any $i\in I$, we have
\begin{equation}\label{eq:specLrotation}
 \text{spec}\left(\Lambda, L(\omega_i)_k\right)=\gamma^{D_i}\text{spec}\left(\Lambda, L(\omega_i)_k\right)
 \,,
\end{equation}
where $L(\omega_i)$ is the $i$-th fundamental representation of $\wg$ and $D_i$
is defined by equation (\ref{Di}).
Hence, the spectrum of $\Lambda$
in the evaluation representation $L(\omega_i)_k$ is invariant under the multiplication by $\gamma^{D_i}$. 

\section{\texorpdfstring{${}^L\mf g^{(1)}$}{L g(1)}-valued connections
and differential equations}\label{sec:L}

Let $\{e_i, h_i,f_i\mid i=0,\dots, n\}\subset {}^L\mf g^{(1)}=\wgr$
be the set of Chevalley generators of ${}^L\mf g^{(1)}$,
and let us denote by $e=\sum_{i=1}^ne_i$.
Let $\widetilde{\mf h}_0\subset\widetilde{\mf g}_{0}$ denote the Cartan subalgebra of the simple Lie algebra
$\widetilde{\mf g}_0$ and let us fix an element $\ell\in\widetilde{\mf h}_0$.
Recall that $h^\vee$ is the dual Coxeter number of $\gu$, as in Table \ref{table:langlands}. Following \cite{FF11} (see also \cite{Sun12, marava15}),
we consider the $^{L}\mf g^{(1)}$-valued connection 
\begin{equation}\label{20141020:eq1}
\mc L(x,E)=\partial_x+\frac{\ell}{x}+e+p(x,E) e_0\,,
\end{equation}
where $p(x,E)=x^{Mh^\vee}-E$, with $M>0$ and $E\in\mb C$.
Let $k\in\mb C$  and introduce the quantities
$$
\omega=e^{\frac{2\pi i}{h^\vee(M+1)}}
\,,
\qquad
\qquad
\Omega=e^{\frac{2\pi iM}{M+1}}=\omega^{h^\vee M}
\,.
$$
The automorphism $\mc{M}_{k}$ of ${}^L\mf g^{(1)}$ defined in Section \ref{sec:liealgebra} can be extended to an automorphism of ${}^L\mf g^{(1)}$-valued connections (leaving $\partial_x$ invariant), which we denote in the same way. Then from equations  \eqref{20141020:eq1} and \eqref{eq:hrelations}
we get 
\begin{equation}\label{20141020:eq4}
\mc{M}_{k}\left(\omega^{k\ad h}\mc L(x,E)\right)
=\omega^k \mc L(\omega^k x,\Omega^kE)
\,.
\end{equation}
We set $\mc L_k(x,E)=\mc M_k\left(\mc L(x,E)\right)$, for every $k\in\mb C$. Let $ \widehat{\mb{C}}$ be the universal cover of $\bb{C}^{*}$. If we consider a family - depending on $E$ -  of solutions $\phi(x,E):\widehat{\mb{C}} \to V_0$  to the (system of) ODE
\begin{equation}\label{20141021:eq1}
\mc L(x,E)\phi(x,E)=0
\,,
\end{equation}
and for $k\in\mb C$ introduce the function
\begin{equation}\label{20150108:eq8}
\phi_k(x,E)=\omega^{-kh}\phi(\omega^kx,\Omega^kE),
\end{equation}
then by equations \eqref{20141020:eq4}, we have that
\beq\label{20141128:eq1}
\mc L_k(x,E)\phi_k(x,E)=0.
\eeq
In other words, $\phi_k(x,E):\widehat{\mb{C}} \to V_k$ is a solution of \eqref{20141021:eq1} for the representation
$V_k$.

\subsection{Fundamental Solutions}\label{sec:wkb}

For any evaluation representation of the Lie algebra $^{L}\mf{g}^{(1)}$, the connection (\ref{20141020:eq1}) defines a linear differential equation with
a Fuchsian singularity at $x=0$ and an irregular singularity at $x=\infty$. We are interested in a solution -- known as fundamental solution -- uniquely specified by a prescribed subdominant (WKB, exponential) behaviour in a Stokes sector containing the positive semiaxis. In order to construct the fundamental solutions we follow \cite{marava15}, where the case of $\mf g$ simply-laced was considered. 

\begin{definition}\label{def:maximal}
Let $A$ be an endomorphism of a vector space $V$.
We say that a eigenvalue $\lambda$ of $A$ is maximal if it is real,  its algebraic multiplicity
is one, and $\lambda > \Re\mu$ for every eigenvalue $\mu$ of $A$.
\end{definition}

We let $V$ be an evaluation representation
of ${}^L\mf g^{(1)}$ such that $\Lambda$ defined by equation \eqref{eq:defLambda} has a maximal eigenvalue $\lambda$. Defining
$\widetilde{\mb{C}}=\mb C\setminus\mb R_{\leq0}$ the complement of the negative real semi-axis in the complex plane, we consider
solutions $\Psi:\widetilde{\mb{C}}\rightarrow V$ of
\beq\label{20141125:eq1}
\mc{L}(x,E)\Psi(x)=\Psi'(x)+ \left(\frac{\ell}{x}+e+p(x,E) e_0\right) \Psi(x) =0
\,.
\eeq
From WKB theory we expect the dominant asymptotic to be proportional to $e^{-\lambda \int^x p(y,E)^\frac{1}{h^{\vee}} } dy$.
Therefore, we need to study the asymptotic expansion of the function $p(x,E)^\frac{1}{h^{\vee}}$: this is of the form
$p(x,E)^\frac{1}{h^{\vee}}=q(x,E)+O(x^{-1-\delta}),$
where
\begin{equation}\label{eq:delta}
\delta=M(h^\vee(1+s)-1)-1>0, \qquad s=\lfloor \frac{M+1}{h^\vee M} \rfloor,
\end{equation}
and
\begin{equation}\label{20150128eq1}
q(x,E)=x^M+\sum_{j=1}^{s} c_j(E) x^{M(1-h^\vee j)}.
\end{equation}
For every $j=1,\dots,s$, the function $c_j(E)$ is a monomial of degree $j$ in $E$. We define the action $S(x,E)$ to be the integral of
$q(x,E)$, and we distinguish two cases.
In the generic case $\frac{M+1}{h^\vee M} \notin \bb{Z}_+$, the action is defined as
\begin{equation}\label{eq:actiongeneric}
S(x,E)=\int^{x}_0 q(y,E) dy \,,
\qquad x \in \widetilde{\mb C}
\,,
\end{equation}
where
we chose the branch
of $q(x,E)$ satisfying $q \sim |x|^{M}$ for $x$ real.
In the case $\frac{M+1}{h^\vee M} \in \bb{Z}_+$ we set
\begin{equation*}
 S(x,E)=\sum_{j=0}^{s-1} \int^{x}_0   c_j(E) y^{M(1-h^\vee j)} dy + c_s(E) \log x \, , \quad s= \frac{M+1}{h^\vee M} \,.
\end{equation*}
 We notice that if $M(h^\vee-1)>1$ the actions $S(x,E)$ coincides with $\frac{x^{M+1}}{M+1}$.
\noindent
We are now in the position to state the main result of this section.%
\begin{theorem}\label{thm:asymptotic}
Let $V$ be a finite dimensional representation of $^{L}\mf{g}^{(1)}$, such that the matrix representing
the action of $\Lambda\in {}^L\mf g^{(1)}$ on $V$ is diagonalizable and has a maximal eigenvalue
$\lambda$. Let  $\psi\in V$ be the corresponding unique (up to a constant) eigenvector. Then,
there exists a unique solution $\Psi(x,E):\widetilde{\mb{C}}\to V$
to equation \eqref{20141125:eq1}
with the following asymptotic behaviour:
$$
\Psi(x,E)
=e^{-\lambda S(x,E)} q(x,E)^{-h} \big( \psi + o(1) \big)\quad \text{ as }\quad x \to +\infty
\,.
$$
Moreover, the same asymptotic behaviour holds in the sector
$|\arg{x}| < \frac{\pi}{2(M+1)} $,
that is, for any $\delta>0$ it satisfies
\begin{equation}\label{20150113:eq1}
\Psi(x,E)
=e^{-\lambda S(x,E)} q(x,E)^{-h} \big( \psi + o(1) \big)
\,,
\quad \text{ in the sector } \,\, |\arg{x}| <\frac{\pi}{2(M+1)} -\delta
\,.
\end{equation}
The function $\Psi(x,E)$ is an entire function of $E$.
\end{theorem}
\begin{proof}
 The proof coincides with the proof of Theorem 2.4 in \cite{marava15}. Indeed, the latter does not depend on the
 choice of the affine Lie algebra, and holds for every finite dimensional representation satisfying the hypotheses of the theorem.
 Without entering into the details, we recall here that the proof is based on two subsequent gauge transformations that bring the original ODE
 into an almost diagonal form. The first transformation is given by $\widetilde{\mc{L}}(x,E)=q(x,E)^{\ad h}\mc{L}(x,E)$, where
\begin{equation*}
\widetilde{\mc{L}}(x,E)=\partial_x+ q(x,E) \Lambda + \frac{\ell- M h}{x}+O(x^{-1-\delta}),
\end{equation*}
and where $\delta$ was defined in \eqref{eq:delta}. For the second, letting
$\ell=\sum_{i\in I}\ell_ih_i$ and $h=\sum_{i\in I}a_ih_i$, and introducing the quantity $N=\sum_{i\in I}(\ell_i-Ma_i)f_i\in\mf {}^L\mf g^{(1)}$,
then we have a transformation of the form
\begin{equation}\label{20151015:eqgauge2}
e^{\alpha(x)\ad N}\widetilde{\mc L}(x,E) =\partial_x + q(x,E)\,\Lambda+O(x^{-1-M})\,,
\end{equation}
where $\alpha(x)=\left(x\,q(x,E)\right)^{-1}$.
Finally after the change of variable $x \to S(x,E)$, equation (\ref{20141125:eq1}) is eventually transformed to 
\begin{equation}\label{eq:quasiconstantode}
\partial_S \Phi(S)  +\big( \Lambda+O(S^{-\frac{1+2M}{1+M}}) \big) \Phi(S)=0\, \; .
\end{equation}
Since, by hypothesis, $\Lambda$ is diagonalizable, then equation (\ref{eq:quasiconstantode}) defines a system of ODEs
in constant diagonalizable form
modulo a $L^1$ remainder.
Standard tools of asymptotic analysis, see e.g.
\cite[Theorem 1.8.1]{Eas89}, show the existence of a basis of solutions of
the form $e^{-\lambda_i S(x,E)}\varphi_i \big(1+o(1)\big), \, x\gg 0$
for any eigenpair $ (\lambda_i,\varphi_i)$ of $\Lambda$. The uniqueness of the subdominant solution $\Psi(x,E)$ is a direct consequence
of the maximality of the eigenvalue
$\lambda$. The extension of the asymptotic formula to the sector $|\arg{x}| <\frac{\pi}{2(M+1)} -\delta$
requires some more work, for which we refer
to the cited Theorem 2.4 in \cite{marava15}.
\end{proof}
\begin{remark}
If $V$ is an evaluation representation or a tensor product of evaluation representations, the matrix representing $\Lambda \in\,\lgu$ in $V$ is diagonalizable,
because $\Lambda$ is a semisimple element \cite{kac90}.
\end{remark}

We conclude this section recalling that for any $k\in\mb R$ such that $|k| < \frac{h^\vee (M+1)}{2}$, the function
\begin{equation}\label{eq:Psik}
 \Psi_k(x,E)=\omega^{-k h} \Psi(\omega^kx,\Omega^kE)\,,
 \quad x \in \bb{R}_+
\end{equation}
 defines, by analytic continuation, a solution $\Psi_k: \widetilde{\mb{C}} \to V_k $
 of equation \eqref{20141125:eq1}.
 Using the expansion (\ref{20150113:eq1}) we obtain that on the positive real semi axis
$$
\Psi_k(x,E)=e^{- \gamma^k\lambda   S(x,E) } q(x,E)^{-h} ( \gamma^{-k h} \psi +o(1))\,,
\qquad x \gg 0 \,,
$$
where $\gamma=e^{\frac{2\pi i}{h^\vee}}$ was defined in (\ref{20151003:eq1}).
%

\section{The \texorpdfstring{$\Psi$}{Psi}-system}\label{sec:psi}
In this section we
prove a system of quadratic relations -- known as  $\Psi$-system -- among $n=\rank{\mf g}$ fundamental solutions $\Psi^{(i)}$ 
of equation (\ref{20141125:eq1}) in certain distinguished evaluation representations of $^{L}\gu$, to be introduced below.
The $\Psi$-system for arbitrary classical Lie algebras has been first conjectured in \cite{dorey07}.
%
\begin{remark}
In our previous paper \cite{marava15}, starting from a simply-laced Lie algebra $\mf g$ we considered a
$\gu$-valued connection as well as suitable evaluation representations of $\gu$,
and we proved the validity of the $\Psi$-system conjectured in \cite{dorey07} for the Lie algebra $\mf g$.
Using that $\Psi$-system, we obtained the $\mf g$-Bethe Ansatz.
We remark that a $\Psi$-system for arbitrary simple Lie algebras $\mf g$ can be obtained 
following precisely the same steps - thus considering a connection with values in $\gu$ rather than in $^{L}\gu$. However, 
it is easy to check that this $\Psi$-system does not lead to the $\mf g$-Bethe Ansatz equations ($Q$-system) if 
the algebra $\mf g$ is non simply-laced.
\end{remark}
We now provide a $\Psi$-system valid for an arbitrary simple Lie algebra $\mf g$, and reducing to the one considered in \cite{marava15} when $\mf g$ is simply-laced. Moreover, we prove in Section \ref{sec:qsystem} that this $\Psi$-system leads to the $\mf g$-Bethe Ansatz.
Let $\mf g$ be a simple Lie algebra and let $(\wg,r)$ be the pair associated to it through the map \eqref{20151022:eq1}. Recall that ${}^L\mf g^{(1)}=\widetilde{\mf g}^{(r)}$, that we denoted by $L(w_i)$, $i\in\widetilde I$,
 the fundamental representations of $\widetilde{\mf g}$, and that for every $i \in \widetilde{I}$ there exists a morphism of  representations of $\widetilde{\mf g}$ defined by equation (\ref{morphism:090715}).
For each $i \in \widetilde{I}$ we consider representations $V^{(i)}$ of ${}^L\mf g^{(1)}$
defined as evaluation representation of the form
\begin{equation}\label{eq:defVi}
 V^{(i)}=L(\omega_i)_{k_i} \,,
\end{equation}
where the numbers $k_i\in\mb C$ will be chosen so that in the representation $V^{(i)}$ the element $\Lambda$ defined
by equation \eqref{eq:defLambda} has a maximal eigenvalue. As proved in Theorem \ref{thm:asymptotic}, the latter condition ensures the existence of a fundamental solution to the ODE \eqref{20141021:eq1} in the evaluation representation $V^{(i)}$.

In order to find the values of the $k_i$'s we proceed as follows.
Given the set $\widetilde{I}$ of the vertices of the Dynkin diagram of $\wg$ numbered as in Table \ref{fig:dynkin}, we introduce a bipartition
\cite{moody87} of the form $\widetilde{I}=\widetilde{I}_1\cup\widetilde{I}_2$ such that $1\in \widetilde{I}_1$ and all edges of the Dynkin diagram of $\wg$ lead
from $\widetilde{I}_1$ to $\widetilde{I}_2$. Then, we define the function $p:\widetilde{I}\longrightarrow \bb{Z}/2\mb Z$ as
$$p(i)=
\begin{cases}
0\quad i\in\widetilde{I}_1,\\
1\quad i\in\widetilde{I}_2.
\end{cases}
$$
It is easy to check using Table \ref{table:graphs} that we always have $p(i)=p(\sigma(i))$. In addition, we set $s_i=(-1)^{p(i)}$, $i\in\widetilde{I}$. We then consider the morphism of $\wg$-modules \eqref{morphism:090715}, which extends to the following morphism of evaluation representations of ${}^L\mf g^{(1)}=\wgr$:
\begin{equation}\label{morphism:090715-2}
\widetilde{m}_i=\bigwedge^2 V^{(i)}_{s_i\frac{D_i}{2}}\longrightarrow \bigotimes_{j\in\widetilde{I}}
V_{k_i-k_j+s_i\frac{D_i}{2}}^{(j)\otimes \widetilde{B}{ij}}.
\end{equation}
Assume now that $\Lambda$ has a maximal eigenvalue in $V^{(i)}$ for each $i \in\widetilde I$. By equation  (\ref{eq:specLrotation}) we expect  $\Lambda$ to have  maximal eigenvalue also in the representation $\bigwedge V^{(i)}_{s_i \frac{D_i}{2}}$.  Requiring in addition $\Lambda$ to have a maximal eigenvalue also in the tensor product representation 
\begin{equation}\label{20151023:eq1}
\bigotimes_{j\in\widetilde{I}}
V_{k_i-k_j+s_i\frac{D_i}{2}}^{(j)\otimes \widetilde{B}{ij}},
\end{equation}
appearing in  \eqref{morphism:090715-2}, provides a way to choose the values of the $k_i$'s. Indeed, it is clear that  $\Lambda$ has a maximal eigenvalue in a representation of the form $\otimes_j V^{(j)}$, provided the indices $j$ appearing in the tensor product belong to different $\sigma$-orbits. Otherwise, since by \eqref{eq:defVi} and Proposition \ref{prop: 120515} ii) we have $V^{(\sigma(j))}\cong V^{(i)}_{-\frac{1}{r}}$, an extra twisting appears. Due to the above argument, and looking at the morphism (\ref{morphism:090715-2}), it seems reasonable to impose the conditions $k_i-k_j+s_i\frac{D_i}{2}=0 $ for any node $i \in \widetilde{I}$ satisfying the property that for every $j\in\widetilde{I}$ such that $j\neq \sigma(j)$ then at least one between  $\widetilde{B}_{ij}$ and $\widetilde{B}_{i\sigma(j)}$ is zero. These are precisely those nodes $i\in\widetilde{I}$  such that at most one node $j$ for each orbit appears in the tensor product representation \eqref{20151023:eq1}. We are led therefore to the following definition.
\begin{definition}\label{def:good}
Let $\sigma$ be a diagram automorphism of $\widetilde{\mf g}$. 
A vertex $i$, $i\in\widetilde I$, of the Dynkin diagram of $\widetilde{\mf g}$ is called \emph{good} (with respect to $\sigma$) if
for every $j\in\widetilde{I}$ such that $j\neq \sigma(j)$, then we have $\widetilde{B}_{ij}\widetilde{B}_{i\sigma(j)}=0$.
\end{definition}
The above condition can be recasted into an equivalent condition on the incidence matrix $B$
of the Dynkin diagram of $\mf g$:
$i\in I$ is good if and only if $B_{ij}\in\left\{0,1\right\}$ for every $j\in I$. 
Moreover, see equation (\ref{Browsum}), if $i,j\in I$ and $i$ is good, then $B_{ij}=\widetilde{B}_{ij}$.
With the above notion of a good vertex of a Dynkin diagram of $\widetilde{\mf g}$
we define inductively the numbers $k_i$ appearing in equation \eqref{eq:defVi}.
\begin{definition}\label{def:ki}
Set $k_1=0$.
If the index $i\in \widetilde{I}$ is good and $j\in \widetilde{I}$ is such that
$\widetilde{B}_{ij}\neq 0$, we define $k_j=k_i+\frac{1}{2} s_i D_i$.
\end{definition}
For any pair $(\widetilde{\mf g},r)$ of a simply-laced Lie algebra $\widetilde{\mf g}$ and diagram automorphism $\sigma$ of order $r$ given by the correspondence \eqref{20151022:eq1}, there is at most one index which is not good. Since the Dynkin diagram is connected,
the inductive procedure in Definition \ref{def:ki} uniquely defines the values of all the $k_i$'s.
We write them explicitly in Table \ref{table:ki}. Note that we always have $k_i=k_{\sigma(i)}$, and that in the case $r=1$ we recover the values of the twists obtained in \cite{marava15}.
\begin{table}[H]\label{table:ki}
\caption{The values of the scalars $k_i$. The number $r$ is the order of $\sigma$.}
{\tabulinesep=1.2mm
\begin{tabu}{ |c|c|c| }
\hline
 $\wg$   & $r$  &   $k_i$, \;$i\in\widetilde{I}$    \\
\hline
\hline
$ADE$  &  $1$ & $k_i=\frac{p(i)}{2}$\\
\hline
$ A_{2n-1}$ & $2$ &
$k_i=\frac{p(i)}{2}$\\
\hline
$ D_{n+1}$,\; $n$ even & $2$ & 
$k_i=\frac{p(i)}{4}$, \;$1\leq i\leq n-1$,\qquad $k_n=k_{n+1}=\frac{1}{2}$ \\
\hline
$ D_{n+1}$,\; $n$ odd &  $2$ &
$k_i=\frac{p(i)}{4}$,\; $1\leq i\leq n-1$,\qquad $k_n=k_{n+1}=-\frac{1}{4}$\\
\hline
$E_6$ &  $2$ & $k_i=\frac{p(i)}{2}$,\; $i\neq 4$,\qquad $k_4=\frac{1}{4}$\\
\hline
$D_4$  &  $3$ & $k_i=\frac{p(i)}{2}$\\
\hline
\end{tabu}
}
\end{table}
The following result will be useful later.
\begin{lemma}\label{lemma130515}
For every $i,j\in I$ such that $B_{ij}\neq 0$, we have
$$
k_j-k_i-\frac{1}{2}s_i D_i=-s_j\frac{B_{ij}-1}{2r}
\,.
$$ 
\end{lemma}
\begin{proof}
If $i$ is good, then $k_j=k_i+\frac{1}{2} s_i D_i$ for every $i,j\in \widetilde{I}$ such that $\widetilde{B}_{ij}\neq 0$. 
Therefore there is nothing to prove since $B_{ij}=\widetilde{B}_{ij}=1$. Suppose now that $i$ is not good. Then, every $j\in \widetilde I\setminus\{ i\}$ is good and we have $k_i=k_j+\frac{1}{2}s_jD_j$
if $\widetilde{B}_{ij}\neq 0$.
Therefore, we get
$$k_j-k_i-\frac{1}{2}s_i D_i=-\frac{1}{2}s_jD_j-\frac{1}{2}s_i D_i=-\frac{1}{2}s_j\left(D_j-D_i\right)\,,$$
where in the last identity we used the fact that $s_i=-s_{j}$ if $\widetilde{B}_{ij}\neq 0$.
From the definition of the matrix $B$,
it follows that if $i$ is not good and $B_{ij}\neq 0$, then $B_{ij}=\langle j \rangle$. This concludes the proof.
\end{proof}
Using Lemma \ref{lemma130515}, for every $i\in I$ we define a morphism $m_i$ of representations constructed using wedge products as well as  tensor products of the $V^{(j)}$ (or their twists), with $j \in I$. This is nothing than the morphism \eqref{morphism:090715-2}, for $i\in I$, and with the choice of the $k_i$'s as in Table \ref{table:ki}.
\begin{proposition}\label{20151019:prop1}
For every $i \in I$, there exists a unique morphism of representations of ${}^L\mf g^{(1)}=\wgr$, given by
 \begin{equation}\label{eq:laverami}
  m_i:\bigwedge^2 V^{(i)}_{\frac{D_i}{2}}\longrightarrow \bigotimes_{j\in I}
  \bigotimes_{\ell=0}^{B_{ij}-1}V^{(j)}_{\frac{B_{ij}-1-2\ell}{2r}}\,, \qquad i\in I,
 \end{equation}
and such that $m_i( f_i v_i \wedge v_i)=\otimes_{j\in I}\;v_j^{\otimes B_{ij}} $, and $\ker m_i=U
$. Here, $v_i$ is a highest weight vector of the fundamental representation $L(\omega_i)$ of $\wg$, and $U$ is the subrepresentation defined in the direct sum decomposition
\eqref{20151024:eq1}. 
\end{proposition}
\begin{remark}
Since $v_i$ is a highest weight vector of the fundamental representation $L(\omega_i)$ of $\widetilde{\mf g}$, then by \eqref{chevg0} we have that $f_i v_i=\widetilde{f}_i v_i$ for every $i\in I$ and $\widetilde{f}_i$
the corresponding Chevalley generator of $\widetilde{\mf g}$.
\end{remark}
\begin{proof}
Recall by equation \eqref{eq:defVi} that $V^{(i)}=L(\omega_i)_{k_i}$ and that
we are assuming $k_i$ as in Definition \ref{def:ki}.
We consider the morphism $\widetilde{m}_i$ given by (\ref{morphism:090715-2}). Due to equation \eqref{Browsum}, the following isomorphisms of ${}^L\mf g^{(1)}$-representations hold:
\begin{equation}\label{20151020:eq1}
\bigotimes_{j\in\widetilde{I}}
V_{k_i-k_j+s_i\frac{D_i}{2}}^{(j)\otimes \widetilde{B}{ij}}
 \cong
\bigotimes_{j\in I}\bigotimes_{\ell=0}^{B_{ij}-1}L(w_{\sigma^l(j)})_{k_i+s_i\frac{D_i}{2}}
\,.
\end{equation}
Furthermore, by Proposition \ref{prop: 120515} ii), we get 
the following isomorphism of representations of ${}^L\mf g^{(1)}$:
\begin{equation}\label{20151020:eq2}
\bigotimes_{j\in I}\bigotimes_{\ell=0}^{B_{ij}-1}L(w_{\sigma^l(j)})_{k_i+s_i\frac{D_i}{2}}
\cong
\bigotimes_{j\in I}\bigotimes_{\ell=0}^{B_{ij}-1}V^{(j)}_{k_i-k_j+s_i\frac{D_i}{2}-\frac{\ell}{r}}
\,.
\end{equation}
Hence,  combining equations \eqref{20151020:eq1} and \eqref{20151020:eq2}, using Lemma \ref{lemma130515} and permuting the terms in the tensor product in the right hand side of
\eqref{20151020:eq2}, we get the following isomorphism of representations:
$$
\xi_i:\bigotimes_{j\in\widetilde{I}}
V_{k_i-k_j+s_i\frac{D_i}{2}}^{(j)\otimes \widetilde{B}{ij}}
\stackrel{\sim}{\rightarrow}
\bigotimes_{j\in I}\bigotimes_{\ell=0}^{B_{ij}-1}V^{(j)}_{\frac{B_{ij}-1-2\ell}{2r}}\,,
\qquad
i\in I
\,.
$$
It can be easily checked that
\begin{equation}\label{20151020:eq3}
\xi_i(\otimes_{j\in\widetilde{I}}v_j^{\otimes \widetilde{B}{ij}})=\otimes_{j\in I} \;v_j^{\otimes B_{ij}}
\,.
\end{equation}
In addition, since $s_i=\pm 1$ and using the isomorphism of representations given by 
equation \eqref{isoR_i} with $k=k_i-\frac{D_i}{2}$, we get the isomorphism
$$
\zeta_i:\bigwedge^2 V^{(i)}_{s_i\frac{D_i}{2}}\stackrel{\sim}{\rightarrow}\bigwedge^2 V^{(i)}_{\frac{D_i}{2}}
\,,
\qquad i\in I
\,. 
$$
Then, for every $i\in I$, the homomorphism in \eqref{eq:laverami} is defined as 
$m_i=\xi_i\circ\widetilde{m}_i\circ\zeta_i^{-1}$, where $\widetilde{m}_i$ is as in equation \eqref{morphism:090715-2}.
Note that $\zeta_i^{-1}(\widetilde f_iv_i\wedge v_i)=\widetilde f_iv_i\wedge v_i$. Hence, by equation
\eqref{20151020:eq3} we get $m_i(\widetilde{f}_iv_i\wedge v_i)=\otimes_{j\in I}v_j^{\otimes B_{ij}}$. This proves the Proposition.
\end{proof}
With the choice of the numbers $k_i$ as in Table \ref{table:ki}, the element $\Lambda$ has a maximal eigenvalue in the representations 
$V^{(i)}$, $i\in I$, defined by equation (\ref{eq:defVi}).
%
%
More precisely, we can prove the following result.
\begin{theorem}\label{thm:genlamba}
Let $\mf g$ be a simple Lie algebra and let $V^{(i)}=L(\omega_i)_{k_i}, i \in I$, with the $k_i$'s as in Table \ref{table:ki}, 
be 
representations of ${}^L\mf g^{(1)}$. Then, we can always normalize $\Lambda$ in such a way that the following facts hold true.
\begin{enumerate}[i)]
 \item In any representation $V^{(i)}, i \in I$, the element $\Lambda$ has a maximal eigenvalue
$\lambda^{(i)}$ and, in particular,  $\lambda^{(1)}=1$.
 We denote by $\psi^{(i)}$ the corresponding unique (up to a constant factor) eigenvector.
\item The following linear relations among the eigenvalues $\lambda^{(i)}$, $i\in I$, hold
 \begin{equation}\label{eq:genlambdarelations}
\mu_i:=  \left(\gamma^{-\frac{D_i}{2}}+\gamma^{\frac{D_i}{2}}\right)\lambda^{(i)}=
  \sum_{j\in I}\Bigg(\sum_{\ell=0}^{B_{ij}-1}\gamma^{\frac{B_{ij}-1-2\ell}{2r}}\Bigg)\lambda^{(j)}
  \,.
 \end{equation}
 
\item
For every $i\in I$, in the representation $\bigotimes_{j\in I}\bigotimes_{\ell=0}^{B_{ij}-1}V^{(j)}_{\frac{B_{ij}-1-2\ell}{2r}}$,
$\Lambda$ has the maximal eigenvalue
$\mu_i$
and the corresponding eigenvector is
$$
\psi_{\otimes}^{(i)}=
\bigotimes_{j\in I}\bigotimes_{\ell=0}^{B_{ij}-1}\psi^{(j)}_{\frac{B_{ij}-1-2\ell}{2r}}\,.
$$

\item In the representation $\bigwedge^2 V^{(i)}_\frac{D_i}{2}$, we have that
$\mu_i$
is a maximal eigenvalue for $\Lambda$ and the corresponding eigenvector is
$$R_i\Big(\psi^{(i)}_{-\frac{D_i}{2}}\Big)\wedge \psi^{(i)}_{\frac{D_i}{2}}.$$

\item We can normalize the eigenvector $\psi^{(i)}$ in such a way that the
following (algebraic) $\Psi$-system holds
\begin{equation}\label{eq:thegloriouspsi}
m_i\Big(R_i\Big(\psi^{(i)}_{-\frac{D_i}{2}}\Big)\wedge\psi^{(i)}_{\frac{D_i}{2}}\Big)=
\bigotimes_{j\in I}\bigotimes_{\ell=0}^{B_{ij}-1}\psi^{(j)}_{\frac{B_{ij}-1-2\ell}{2r}}\,, \qquad i\in I. 
\end{equation}
\end{enumerate}
\end{theorem}
If $\mf g$ is simply-laced, then ${}^L\mf g^{(1)}$ is an untwisted affine Kac-Moody algebra
and $\Lambda$ defined by \eqref{eq:defLambda} can be characterized as an eigenvector of a Killing-Coxeter transformation.
The proof of Theorem \ref{thm:genlamba} in the simply-laced case was obtained in \cite{marava15} using this fact and a related result by Kostant \cite{kos59}.
For $\mf g$ non simply-laced, ${}^L\mf g^{(1)}$ is a twisted affine Kac-Moody algebra and some of the properties of the 
element $\Lambda$ used in the proof in the simply-laced case fail to hold.
In particular, $\Lambda$ is not the eigenvector of a Killing-Coxeter transformation
(nor of a twisted Killing-Coxeter transformation), and the study of its spectrum
in the fundamental representations $V^{(i)}$ of ${}^L\mf g^{(1)}$  requires a different approach.
We prove Theorem \ref{thm:genlamba} for a non simply-laced Lie algebra $\mf g$ --  by a direct case by case inspection --  in Section \ref{sec:proof_main}.\\

By Theorem \ref{thm:genlamba}(i), for any representation $V^{(i)}$, $i \in I$,
there exists a maximal eigenvalue $\lambda^{(i)}$ and a maximal eigenvector $\psi^{(i)}$. It follows from Theorem \ref{thm:asymptotic} that there exists a fundamental solution
$\Psi^{(i)}(x,E):\widetilde{\bb{C}} \to V^{(i)}$ with the following asymptotic behavior
\begin{equation}\label{eq:Psii}
 \Psi^{(i)}(x,E)
 =e^{-\lambda^{(i)} S(x,E)} \big( \psi^{(i)} + o(1) \big)
 \,,
\text{ in the sector } |\arg{x}| <\frac{\pi}{2(M+1)}
\,.
\end{equation}
We are now in the position to establish the $\Psi$-system.
\begin{theorem}\label{thm:psi-sistem}
Let $\mf g$ be a simple Lie algebra and let the solutions $\Psi^{(i)}(x,E):\widetilde{\mb C}\to V^{(i)}$,
$i\in I$, have the asymptotic behaviour \eqref{eq:Psii}.
Then, the following identity, known as $\Psi$-system, holds for every $i\in I$:
\begin{equation}\label{eq:031201}
m_i\left(  R_i \big(\Psi_{-\frac{D_i}{2}}^{(i)}(x,E) \big) \wedge
\Psi_{\frac{D_i}{2}}^{(i)}(x,E) \right) =\bigotimes_{j\in I}\bigotimes_{\ell=0}^{B_{ij}-1}\Psi^{(j)}_{\frac{B_{ij}-1-2\ell}{2r}}(x,E)
\,.
\end{equation}
Here, the morphism $m_i$ is defined by equation (\ref{eq:laverami}) and the isomorphism $R_i$ is defined 
by equation (\ref{isoR_i}).
\end{theorem}
\begin{proof}
Due to Theorem \ref{thm:genlamba} iii) and Theorem \ref{thm:asymptotic},
the unique subdominant solution
to equation \eqref{20141125:eq1} in the representation $\bigotimes_{j\in I}\bigotimes_{\ell=0}^{B_{ij}-1}V^{(j)}_{\frac{B_{ij}-1-2\ell}{2r}}$ 
is
$$
\bigotimes_{j\in I}\bigotimes_{\ell=0}^{B_{ij}-1}\Psi^{(j)}_{\frac{B_{ij}-1-2\ell}{2r}}(x,E)
=e^{- \mu^{(i)}S(x,E)} q(x,E)^{-h}\left( \psi_\otimes^{(i)}+o(1)\right),\quad x\gg0\,,
$$
where $\mu^{(i)}= \sum_{j\in I}\left(\sum_{\ell=0}^{B_{ij}-1}\gamma^{\frac{B_{ij}-1-2\ell}{2r}}\right)\lambda^{(j)}$.
Moreover, by equation \eqref{eq:Psik},
and Theorem \ref{thm:genlamba} iv) we have that
$$R_i \big(  \Psi_{-\frac{D_i}{2}}^{(i)}(x,E) \big) \wedge\Psi_{\frac{D_i}{2}}^{(i)}(x,E)=e^{- \mu^{(i)}S(x,E)} q(x,E)^{-h}  \left( R_i \big( \psi_{-\frac{D_i}{2}}^{(i)} \big) \wedge\psi_{\frac{D_i}{2}}^{(i)}  +o(1)\right),$$
for $x\gg0$. The proof follows by Theorem \ref{thm:genlamba} i) and v) and the uniqueness of the subdominant solution.
\end{proof}
\begin{example}\label{exa:oldpsi}
For a simply-laced Lie algebra $\mf g$ equation \eqref{eq:031201} becomes 
($i=1,\ldots,n=\rank\mf g$)
\begin{equation}\label{psisystem:old}
m_i\left( \Psi_{-\frac{1}{2}}^{(i)}(x,E)  \wedge
\Psi_{\frac{1}{2}}^{(i)}(x,E) \right) =\bigotimes_{j\in I}\Psi^{(j)}(x,E)^{\otimes B_{ij}}
\,,
\end{equation}
where $I=\{1,\ldots,n\}$ and $B=(B_{ij})_{i,j=1}^n$
is the incidence matrix of the Dynkin diagram of $\mf g$ (see Proposition \ref{prop:gfromgtilde}).
This $\Psi$-system coincides with the one obtained in \cite{marava15}.
\end{example}
\begin{example}\label{exa:Bn}
For $\mf g$ of type $B_n$, $n\geq3$, equation \eqref{eq:031201} becomes
\begin{align}
\begin{split}\label{psisystem:Bn}
&m_i\left(  \Psi_{-\frac12}^{(i)} \wedge \Psi_{\frac12}^{(i)} \right)
=\Psi^{(i-1)}\otimes\Psi^{(i+1)}
\,,
\qquad i=1,\dots,n-1\,,
\\
&m_{n}\left(  R_{n} \left(\Psi_{-\frac14}^{(n)} \right) \wedge \Psi_{\frac14}^{(n)} \right)
=\Psi_{\frac14}^{(n-1)}\otimes\Psi_{-\frac14}^{(n-1)}
\,,
\end{split}
\end{align}
where we set $\Psi^{(0)}=1$.
\end{example}
\begin{example}\label{exa:Cn}
For $\mf g$ of type $C_n$, $n\geq2$, equation \eqref{eq:031201} becomes
\begin{align}
\begin{split}\label{psisystem:Cn}
&m_i\left(  R_i \left(\Psi_{-\frac14}^{(i)} \right) \wedge \Psi_{\frac14}^{(i)} \right)
=\Psi^{(i-1)}\otimes\Psi^{(i+1)}
\,,
\qquad i=1,\dots,n-2\,,
\\
&m_{n-1}\left(  R_{n-1} \left(\Psi_{-\frac14}^{(n-1)} \right) \wedge \Psi_{\frac14}^{(n-1)} \right)
=\Psi^{(n-2)}\otimes\Psi_{\frac14}^{(n)}\otimes\Psi_{-\frac14}^{(n)}
\,,
\\
&m_{n}\left(  \Psi_{-\frac12}^{(n)} \wedge \Psi_{\frac12}^{(n)} \right)
=\Psi^{(n-1)}
\,,
\end{split}
\end{align}
where we set $\Psi^{(0)}=1$.
\end{example}
\begin{example}\label{exa:F4}
For $\mf g$ of type $F_4$, equation \eqref{eq:031201} becomes
\begin{align}
\begin{split}\label{psisystem:F4}
&m_i\left(  \Psi_{-\frac12}^{(i)}  \wedge \Psi_{\frac12}^{(i)} \right)
=\Psi^{(i-1)}\otimes\Psi^{(i+1)}
\,,
\qquad i=1,2\,,
\\
&m_{3}\left(  R_{3} \left(\Psi_{-\frac14}^{(3)} \right) \wedge \Psi_{\frac14}^{(3)} \right)
=\Psi^{(2)}_{\frac14}\otimes\Psi_{-\frac14}^{(2)}\otimes\Psi^{(4)}
\,,
\\
&m_{4}\left(  R_{4} \left(\Psi_{-\frac14}^{4)} \right) \wedge \Psi_{\frac14}^{(4)} \right)
=\Psi^{(3)}
\,,
\end{split}
\end{align}
where we set $\Psi^{(0)}=1$.
\end{example}
\begin{example}\label{exa:G2}
For $\mf g$ of type $G_2$, equation \eqref{eq:031201} becomes
\begin{align}
\begin{split}\label{psisystem:G2}
&m_1\left(  \Psi_{-\frac12}^{(1)}\wedge \Psi_{\frac12}^{(1)} \right)
=\Psi^{(2)}
\,,
\\
&m_{2}\left(  R_{2} \left(\Psi_{-\frac16}^{(2)} \right) \wedge \Psi_{\frac16}^{(2)} \right)
=\Psi^{(1)}_{\frac13}\otimes\Psi^{(1)}\otimes\Psi^{(1)}_{-\frac13}
\,.
\end{split}
\end{align}
\end{example}

For every simple Lie algebra, the $\Psi$-system provided in the examples above coincides with the one conjectured in \cite{dorey07} (for the highest weight
component of the vectors $\Psi^{(i)}(x,E)$).
In this direction some partial results were already obtained by \cite{Sun12} in the case of classical Lie algebras.

\begin{remark}
The authors of the paper \textit{ODE/IM correspondence and Bethe Ansatz equations for affine Toda field equations}, Nucl. Phys. B. 896 (2015), whose arXiv version appeared after \cite{marava15}, claim to have obtained the $\Psi-$system for an arbitrary simple Lie algebra.  However, that paper contains no derivation of the $\Psi-$system, for instance, the representations $V^{(i)}$ are not even defined. In addition, except that in the case $A_n^{(1)}$ (which was already known, \cite{dorey07, marava15})  the $\Psi-$system proposed is  not correct. We point out some major inconsistencies
in the proposed $\Psi-$system.
First the authors claim that there is an embedding $\iota$ (that should correspond to our morphism $\widetilde{m}_i$) of
$\bigwedge^2 L(\omega_i)$ into $\bigotimes_{j\in\widetilde{I}}L(\omega_j)^{\otimes \widetilde{B}{ij}}$. Such embedding does not need to exist. In fact,
for example
in the case $D_n, n>5$, the dimension of $\bigwedge^2 L(\omega_n)$ is larger than the dimension of $L(\omega_{n-2})$. Second,
the alleged $\Psi-$system is written 
using expressions
of the kind $\iota(\Psi_{-\frac14}\wedge \Psi_{\frac14})$, which are meaningless as $\Psi_{-\frac14}$ and $\Psi_{\frac14}$
belong to distinct representations and thus
the action of the algebra on their wedge product is not defined.
Third, the asymptotic behavior of the $\Psi$ function, equation (2.22) or (A.3) of the above mentioned paper,
is valid only if $M (h^\vee-1) >1$ or $E=0$, in which case it coincides with the asymptotic behaviour already proved
(in general) in \cite{marava15} (see also Theorem \ref{thm:asymptotic} above).
We finally note that the important relations among maximal eigenvalues of $\Lambda$ in different representations,
equation $(3.8)$ of the above mentioned paper, was verified there \lq\lq for many cases\rq\rq . However, a complete proof of these
identities was already available \cite[Proposition 3.4 a)]{marava15}.
\end{remark}

    \section{Proof of Theorem \ref{thm:genlamba}}\label{sec:proof_main}
    We first state some preliminary results, that will be used in the proof of Theorem \ref{thm:genlamba}.
    The following result is well-known.
    \begin{lemma}\label{lem:pisa13gen}
    Let $A$ be an endomorphism of a vector space $V$ with eigenvalues
    $\lambda_1,\lambda_2,\dots\lambda_r$.
    Let us assume that
    $$
    \Re \lambda_1\geq\Re\lambda_2\geq\dots\geq\Re \lambda_k>
    \Re \lambda_{k+1}\geq\Re\lambda_{k+2}\geq\dots\geq\Re \lambda_r
    \,.
    $$
    Then,
    $$
    \mu=\sum_{j=1}^k\lambda_j
    $$
    is a maximal eigenvalue for the action of $A$ on $\bigwedge^kV$,
    provided that $\mu\in\mb R$.
    \end{lemma}
\begin{lemma}\label{lem:tracemorphism}
Let $V=\bb{C}^n$ be the standard representation of $\mf{sl}_n$. For every $i=1,\dots,n-1$, there exists an embedding (that is, an injective morphism) of representations
\begin{equation}\label{eq:genembedding}
\bigwedge^2 \Big(\bigwedge^i V\Big) \hookrightarrow  \Big(\bigwedge^{i-1} V\Big) \otimes \Big(\bigwedge^{i+1}V\Big),
\end{equation}
where $\bigwedge^0 V \cong \bigwedge^n V \cong \bb{C}$ is the trivial one-dimensional module.
\end{lemma}
\begin{proof} Recall the standard decomposition
$$\Big(\bigwedge^iV\Big)\otimes\Big(\bigwedge^i V\Big)=\text{Sym}^2\Big(\bigwedge^i V\Big)\oplus \bigwedge^2\Big(\bigwedge^i V\Big).$$
By a special case of  \cite[Exercise 15.30]{FH91} there is a natural homomorphism of representations
$$\varphi_i:\Big(\bigwedge^iV\Big)\otimes\Big(\bigwedge^i V\Big)\rightarrow \Big(\bigwedge^{i-1} V\Big) \otimes \Big(\bigwedge^{i+1}V\Big)\,,$$
explicitly given by
\begin{align*}
&  (v_1\wedge\dots\wedge v_i)\otimes(w_1\wedge\dots\wedge w_i)\\
&\qquad \mapsto \sum_{k=1}^{i}(-1)^k(v_1\wedge\dots\wedge\widehat{v}_k\wedge\dots\wedge v_i)\otimes(v_k\wedge w_1\wedge\dots\wedge w_i),
\end{align*}
and the kernel of $\varphi_i$ is the highest weight representation with highest weight $2\,\omega_i$, where $\omega_i$ the $i$-th fundamental weight. Since $2\,\omega_i$ is a highest weight vector of $\text{Sym}^2\Big(\bigwedge^i V\Big)$,  it follows that $\ker \varphi_i\subset \text{Sym}^2\Big(\bigwedge^i V\Big)$
and this implies the existence of the embedding \eqref{eq:genembedding}.
\end{proof}
 \begin{lemma}\label{lem:vadermonde}
    Let $\mf g$ be a simple Lie algebra
    and let $V$ be an irreducible finite dimensional representation.
    Let $\lbrace u_1 ,\dots,u_N \rbrace $ be a basis of $V$ consisting of weight vectors such that
    $u_1$ is the highest weight vector,
    and we denote by $\Ht(i)$ the height of the weight vector $u_i$, for $i=1,\dots,N$.
    Let $\psi^{(1)}=\sum_{i=1}^N c_i u_i\in V$ and $\psi_k^{(1)} =\gamma^{- h k } \psi$, $k\in\mb C$.
    For every $l \in \bb{C}$ and $M\leq N$ define
    $$\psi^{(M)}_l=
    \psi_{l}^{(1)}\wedge\psi_{l+1}^{(1)}
    \wedge\dots\wedge\psi_{l+M-2}^{(1)}\wedge\psi_{l+M-1}^{(1)}
    \in\bigwedge^M V \,,
    $$
    and write
    $$
    \psi^{(M)}_l=\sum_{1\leq i_1<\dots<i_M\leq N}\alpha_{i_1,\dots,i_M}u_{i_1}\wedge\dots\wedge u_{i_M}
    \,.
    $$
    If $F:\lbrace 1,\dots, M\rbrace \to \lbrace 1,\dots, N\rbrace$ is a strictly increasing function,
    then $\alpha_{F(1),\dots,F(M)}=0$ only if
    at least one of the two following conditions hold:
    \begin{enumerate}
     \item[\text{(a)}] $c_{F(i)}=0$ for some $i=1,\dots,M$;
     \item[\text{(b)}] $\Ht(F(i)) \equiv\Ht(F(j)) \bmod h^\vee$  for some $i \neq j$.
    \end{enumerate}
    \end{lemma}
    \begin{proof}
    Since $u_1$ is a highest weight vector we have that  $h u_1=\alpha u_1$, for some $\alpha\in\mb C$.
    For every $i\in \widetilde{I}$, we have the commutation relation $[h,f_i]=-f_i$ (see \eqref{eq:hrelations}),
    thus we get $hu_j=(\alpha-\Ht(j))u_j$, for every $j=1,\dots,N$.
    Hence, $\gamma^{-kh}u_j=\gamma^{k(\Ht(j)-\alpha)}u_j$, from which follows that
    $$
    \psi_l^{(1)}=\gamma^{-\alpha l}\sum_{i=1}^Nc_i\gamma^{\Ht(i)l}u_i
    \,.
    $$
    By elementary linear algebra we get
    $$
    \alpha_{F(1),\dots,F(M)}=\Big(\prod_{i=1}^M c_{F(i)}\Big)
    \det \big( \gamma^{(\Ht(j)-\alpha)(l+i-1)}\big)_{i,j=1}^{M}
    \,.
    $$
    Using the Vandermonde determinant formula the right hand side above becomes
    $$
    \gamma^{-\alpha\frac{M(M-1)}{2}} \Big(\prod_{i=1}^M c_{f(i)}\Big)
    \Big(\prod_{i=1}^M \gamma^{(\Ht(F(i))-\alpha)(l-1)}\Big)
    \prod_{i<j}\big(\gamma^{\Ht(F(i))}-\gamma^{\Ht(F(j))}\big)
    \,.
    $$
    This term vanishes when one of the conditions (a) or (b) is satisfied thus concluding the proof.
    \end{proof}
    %

    \subsection{Proof of Theorem \ref{thm:genlamba} for $\mf g$ of type $B_n$, $n\geq 3$.}\label{sec:proof_Bn}
    The dual Coxeter number of $\mf g^{(1)}$ is $h^\vee=2n-1$
    thus we set $\gamma=e^{\frac{2\pi i}{2n-1}}$.
    Recall that in this case ${}^L\mf g^{(1)}$ is of type $A_{2n-1}^{(2)}$.
    The simple Lie algebra of type $A_{2n-1}$, $n\geq 1$, can be realized as the algebra of $(2n)\times(2n)$
    traceless matrices
    $$
    \widetilde{\mf g}=\mf{sl}_{2n}=\{A\in\Mat_{2n}(\mb C)\mid \tr A=0\}
    \,,
    $$
    where the Lie bracket is the usual commutator of matrices.
    Let us consider the following Chevalley generators of $\wg$ ($i\in \widetilde{I}=\{1,\dots,2n-1\}$):
    $$
    \widetilde{f}_i=E_{i+1,i}\,,
    \qquad
    \widetilde{h}_i=E_{ii}-E_{i+1,i+1}\,,
    \qquad
    \widetilde{e}_i=E_{i,i+1}
    \,,
    $$
    where $E_{ij}$ denotes the elementary matrix with $1$ in position $(i,j)$ and $0$ elsewhere.
    It is well-known that the representation $L(\omega_1)$ is given by the natural
    action of $\wg$ on $L(\omega_1)=\mb C^{2n}$.
    Moreover, we have that
\begin{equation}\label{20151029:eq1}
    L(\omega_i)=\bigwedge^i L(\omega_1)\,,
    \qquad
    i\in \widetilde{I}\,.
\end{equation}
    We denote by $u_j$, $j=1,\dots,2n$, the standard basis
    of $\mb C^{2n}$, and we have that
    $$v_i=u_1\wedge u_2\wedge \dots\wedge u_i, \qquad i\in \widetilde{I},$$
    is a highest weight vector of the representation $L(\omega_i)$.
    Using the numbers $k_i$'s for $\widetilde{\mf g}$ given by Table \ref{table:ki}, one gets
\begin{equation}\label{eq:20151103eq1a}
V^{(i)}=\bigwedge^i L(\omega_1)_{\frac{i-1}{2}},\qquad i\in \widetilde{I}.
\end{equation}
    By Lemma  \ref{lem:tracemorphism} and looking at the values of the $D_i$'s as in Table \ref{table:graphs}, we have the following morphisms of representations
    of $\wad$
    \begin{align}
    \begin{split}\label{eq:embeddingsA2n-1}
    & \bigwedge^2 V^{(i)}_{\frac12} \hookrightarrow V^{(i-1)} \otimes V^{(i+1)} \,,
    \qquad i=1,\dots,n-1\,, \\
    & \bigwedge^2 V^{(n)}_{\frac14} \hookrightarrow V^{(n-1)}_{-\frac14} \otimes V^{(n-1)}_{\frac14}
    \,.
    \end{split}
    \end{align}
In the last formula we used the relation $L(\omega_{n+1})_k \cong L(\omega_{n-1})_{k +\frac12}$, which was proved in Proposition \ref{prop: 120515} ii).
A set of Chevalley generators for $A_{2n-1}^{(2)}$ can be obtained as follows ($i=1,\dots,n-1$):
\begin{align*}
&
e_0=\frac12(E_{2n-1,1}+E_{2n,2})t\,,
&&
e_{i}=\widetilde e_{i}+\widetilde e_{2n-i}
\,,
&&
e_n=\widetilde{e}_{n}\,,
    \\
    &
    h_0=E_{2n-1,2n-1}+E_{2n,2n}
    &&
    h_{i}=\widetilde{h}_i+\widetilde{h}_{2n-i}
    \,,
    &&
    h_n=\widetilde{h}_n
    \,,
    \\
    &
    -E_{1,1}-E_{2,2}+2c\,,
    &&
    &&
    \\
    &
    f_0=2(E_{1,2n-1}+E_{2,2n})t^{-1}\,,
    &&
    f_{i}=\widetilde{f}_i+\widetilde{f}_{2n-i}
    \,,
    &&
    f_n=\widetilde{f}_n
    \,.
    \end{align*}
    Recall that $\Lambda=e_0+e_1+\dots+e_n$. We set
    \begin{equation}\label{app:psi_1}
    \psi^{(1)}
    =2\sum_{j=1}^{2n-1}u_j + u_{2n}
    \in V^{(1)}
    \,.
    \end{equation}
    Then, it is easy to check that $\Lambda \psi^{(1)}=\psi^{(1)}$.
    By equation \eqref{20151003:eq1}, for every $j=0,\dots,2n-2$, we have that
    $\gamma^{-jh}\psi^{(1)}\in V^{(1)}$ is an eigenvector
    with eigenvalue $\gamma^j$.
    Since the element $\Lambda$ is semi-simple, a simple dimensionality argument
    shows that  $0$ is an eigenvalue of multiplicity one.
    Hence, for the representation $V^{(1)}$, $\lambda^{(1)}=1$ is a maximal eigenvalue
    with corresponding eigenvector $\psi^{(1)}$.
Furthermore, by Lemma \ref{lem:pisa13gen}, and due to \eqref{eq:20151103eq1a}, it follows that for every $i\in I$,
\begin{equation}\label{eq:lambda_A}
\lambda^{(i)}=\gamma^{-\frac{i-1}2}\sum_{j=0}^{i-1}\gamma^j=\frac{\sin\left(\frac{i\pi}{2n-1}\right)}{\sin\left(\frac{\pi}{2n-1}\right)}\,,
\end{equation}
is a maximal eigenvalue of $\Lambda$ in the representation $V^{(i)}$.
This proves part i).
Note that, in this particular case, equations \eqref{eq:genlambdarelations} become
\begin{equation}\label{eq:a2n-1lambdarelations}
(\gamma^{-\frac12}+\gamma^{\frac{1}{2}})\lambda^{(i)}=\lambda^{(i-1)}+\lambda^{(i+1)},
\qquad i =1,\dots,n -1 \, , \qquad \lambda^{(n)}=\lambda^{(n-1)},
\end{equation}
where we set $\lambda^{(0)}=0$.
Using \eqref{eq:lambda_A}, one checks directly that equations \eqref{eq:a2n-1lambdarelations} are satisfied, thus proving part ii).

 We prove part iii) and part iv) in the case $i\neq n$ first.
    In this case the representation
     $\bigotimes_{j\in I}\bigotimes_{\ell=0}^{B_{ij}-1}V^{(j)}_{\frac{B_{ij}-1-2\ell}{2r}}$ is simply $V^{(i-1)} \otimes V^{(i+1)}$,
     and by part i), it has the maximal eigenvalue $\lambda^{(i-1)}+\lambda^{(i+1)}$ with eigenvector
    $\psi_{\otimes}^{(i)}=\psi^{(i-1)}\otimes \psi^{(i+1)}$. This proves part iii).
    Part iv) follows from part iii).
    Indeed, using the embeddings (\ref{eq:embeddingsA2n-1})
    we can prove that $\lambda^{(i-1)}+\lambda^{(i+1)}$ either is the maximal eigenvalue of
    $\bigwedge^2 V^{(i)}_{\frac12}$ or it is not an eigenvalue of the latter representation.
    Clearly, $\gamma^{-\frac12 h} \psi^{(i)} \wedge \gamma^{\frac{1}{2}h}\psi^{(i)} $ belongs to $\bigwedge^2 V^{(i)}_{\frac12}$ and it is an eigenvector with
    eigenvalue $(\gamma^{-\frac12 }+ \gamma^{\frac{1}{2}} )\lambda^{(i)}  $ which is equal to $\lambda^{i-1}+\lambda^{i+1}$ by (\ref{eq:a2n-1lambdarelations}).

Let us prove now parts iii) and iv) for $i=n$.
Using the same argument as in the case $i\neq n$, part iv) follows directly from part iii),
which we prove here. Clearly $\gamma^{\pm \frac14 h} \psi^{(n-1)}$ belongs to
    $V^{(n-1)}_{\pm \frac14}$ and it is an eigenvector with eigenvalue $\gamma^{\pm \frac14}\lambda^{(n-1)}$.
    To conclude the proof it is sufficient to show
    that $\gamma^{\pm \frac14}\lambda^{(n-1)}$ is the unique eigenvalue with maximal real part of $V^{(n-1)}_{\pm \frac14}$.
We prove this by contradiction. Suppose that there is another eigenvector $\phi_{\pm}$ (not proportional to $\gamma^{\pm \frac14 h} \psi^{(n-1)}$) with eigenvalue $x_{\pm}$ such that
    $\Re x_{\pm}\geq \Re\gamma^{\pm \frac14}\lambda^{(n-1)} $. Note that $|x_{\pm}| \leq |\lambda^{(n-1)}|$ since
    $|\lambda^{(n-1)}|$ maximize the modulus
    of the eigenvalues of $V^{(n-1)}_k$ for any $k \in \bb{R}$. Therefore  $ \Re \gamma^{\pm \frac14}x_{\pm} \geq \Re\gamma^{\pm \frac12}\lambda^{(n-1)} $.
    Since $\gamma^{\pm \frac14 h} \varphi_{\pm}$ belongs to $ V^{(n-1)}_{\frac12}$, it follows that
    $(\gamma^{\frac12}+\gamma^{-\frac12})\lambda^{(n-1)}$ is not the (unique) maximal eigenvalue of $\bigwedge^2 V^{(n-1)}_{\frac12} $.
    This contradicts part iv) for $i\neq n$.

Finally, we  prove part v).
Recall by part i) that for any eigenvalue $\lambda$ of $\Lambda$ in $V^{(i)}$, we have $|\lambda|\leq \lambda^{(i)}$, and using equations \eqref{eq:hrelations} it can be easily checked that the eigenvector corresponding to $\lambda^{(i)}$ is given by
$$
\psi^{(i)}=\psi_{-\frac{i-1}2}^{(1)}\wedge\psi_{-\frac{i-3}2}^{(1)}\wedge\dots\wedge\psi_{\frac{i-3}2}^{(1)}\wedge\psi_{\frac{i-1}2}^{(1)}\in V^{(i)}$$
Moreover, noticing that $\Ht(j)=j-1$, for $j=1,\dots,2n$, (in fact $u_{j+1}=f_{j}u_{j}$),
we can use use Lemma \ref{lem:vadermonde}, with $F(j)=j$, for $1\leq j\leq i$, and
    the explicit expression for $\psi^{(1)}$ given in
    equation \eqref{app:psi_1} to get
    \begin{equation}\label{app:A_4}
    \psi^{(i)}= c_i v_i+ y_i
    \,,
    \end{equation}
    where $c_i \neq 0$ and $y_i$ is a combination of lower weight vectors.
    By the decomposition \eqref{app:A_4}, we have that the vector
    $$
    \psi_{\otimes}^{(i)}\in
    \bigotimes_{j\in I}\bigotimes_{\ell=0}^{B_{ij}-1}V^{(j)}_{\frac{B_{ij}-1-2\ell}{2r}}
    $$
    has a non-trivial component in the highest weight subrepresentation generated by the highest weight vector
    $ \otimes_{j\in I}v_j^{B_{ij}}$.
    By part iii), the corresponding eigenvalue has multiplicity one and therefore $\psi_{\otimes}^{(i)}$
    belongs to this highest weight subrepresentation, which coincides with the image of the morphism $m_i$.
    By part i), we have that
    $R_i\left(\psi^{(i)}_{-\frac{D_i}{2}}\right)\wedge \psi^{(i)}_{\frac{D_i}{2}}$ is an eigenvector of $\Lambda$ with the same eigenvalue as
    $\psi_\otimes^{(i)}$. Since this eigenvalue is unique by part iv), we have that
    \begin{equation*}
    m_i\left(R_i\left(\psi^{(i)}_{-\frac{D_i}{2}}\right)\wedge\psi^{(i)}_{\frac{D_i}{2}}\right)= \beta_i
    \psi_{\otimes}^{(i)}
    \,, \qquad i\in I\,.
    \end{equation*}
    for some non-zero $\beta_i$. We can always normalize all the $\psi^{(i)}$'s in order to obtain $\beta_i=1$
    for all $i \in I$.

    \subsection{Proof of Theorem \ref{thm:genlamba} for $\mf g$ of type $C_n$, $n\geq 2$.}\label{sec:proof_Cn}
    The dual Coxeter number of $\mf g^{(1)}$ is $h^\vee=n+1$
    thus we set $\gamma=e^{\frac{2\pi i}{n+1}}$.
    Recall that in this case we have ${}^L\mf g^{(1)}=D_{n+1}^{(2)}$.
    Let $n\in\mb Z_+$, and consider the involution on the set
    $\{1,\dots,2n+2\}$ defined by $i\to i^\prime=2n+3-i$.
    Given a matrix $A=\left(A_{ij}\right)_{i,j=1}^{2n+2}\in\Mat_{2n+2}(\mb C)$ we define its anti-transpose
    (the transpose with respect to the antidiagonal) by
    $$
    A^{\at}=\left(A_{ij}^{\at}\right)_{i,j=1}^{2n}\,,
    \quad
    \text{where}
    \quad A_{ij}^{\at}=A_{j^\prime i^\prime}
    \,.
    $$
Let us set
$$S=\sum_{k=1}^{n+1}(-1)^{k+1}\left(E_{kk}+E_{k^\prime k^\prime} \right)\,. $$
    Following \cite{DS85}, the simple Lie algebra of type $D_{n+1}$ can be realized as the algebra
    $$
    \widetilde{\mf g}=\mf o_{2n+2}=\{A\in\Mat_{2n+2}(\mb C)\mid AS+SA^{\at}=0\}
    \,,
    $$
    where the Lie bracket is the usual commutator of matrices.
    For $1\leq i,j\leq 2n+2$, we define
    $$
    F_{ij}=E_{ij}+(-1)^{i+j+1}E_{j^\prime i^\prime}\,,
    \qquad
    G_{ij}=E_{ij}+(-1)^{i+j}E_{j^\prime i^\prime}
    \,,
    $$
    and we consider the following Chevalley generators of $\mf g$:
    %
    \begin{align*}
    &\widetilde{f}_i=F_{i+1,i}\,,
    &&\widetilde{h}_i=F_{ii}-F_{i+1,i+1}\,,
    &&\widetilde{e}_i=F_{i,i+1}\,,
    &\!i=1,\dots,n\,,
    \\
    &\widetilde{f}_{n+1}=2G_{(n+1)^\prime,n}\,,
    &&\widetilde{h}_{n+1}=F_{n,n}+F_{n+1,n+1}\,,
    &&\widetilde{e}_{n+1}=\frac12G_{n,(n+1)'}
    \,.
    \end{align*}
    It is well-known that the representation $L(\omega_1)$ is given by the natural
    action of $\wg$ on $L(\omega_1)=\mb C^{2n+2}$.
    Moreover we have that
    $$
    L(\omega_i)=\bigwedge^iL(\omega_1)\,,\qquad i=1,\dots,n-1\,,
    $$
    while $L(\omega_{n})$ and $L(\omega_{n+1})$ are the so-called half-spin representations of $\widetilde{\mf g}$.
    For convenience we introduce the irreducible but not fundamental representation $ \bigwedge^nL(\omega_1)\cong L(\omega_n+\omega_{n+1})$.
    Since the algebra $D_{n+1}$ is a subalgebra of $\mf{sl}_{2n+2}$,
    by Lemma \ref{lem:tracemorphism} we have the embeddings
    \begin{align}\label{eq:embeddingDn}
    & \bigwedge^2 L(\omega_i) \hookrightarrow L(\omega_{i-1}) \otimes L(\omega_{i+1}) \, ,\quad i=1,\dots,n-2\,,
    \\
    & \bigwedge^2 L(\omega_{n-1}) \hookrightarrow L(\omega_{n-2}) \otimes \bigwedge^n L(\omega_1) \,,
    \label{eq:embeddingDn_2}
    \end{align}
    where $L(\omega_0) \cong  \bb{C}$.
    Moreover, we have the following decompositions  (see \cite{FH91} and \cite[Appendix A.2]{marava15})
    \begin{align}
    &L(\omega_n) \otimes L(\omega_{n+1})\cong \left(\bigwedge^nL(\omega_1)\right)\oplus L(\omega_{n-2})\oplus  L(\omega_{n-4}) \oplus \dots\,,
    \label{eq:embeddingDn_3}\\
    & \bigwedge^2 L(\omega_n) \cong \bigwedge^2 L(\omega_{n+1})
    \cong L(\omega_{n-1})\oplus L(\omega_{n-5})\oplus L(\omega_{n-9})\oplus\dots\,
    \label{eq:embeddingDn_4}\,.
    \end{align}
    We denote by $u_j$, $j=1,\dots,2n+2$, the standard basis
    of $\mb C^{2n+2}$, and by $v_i$ the highest weight vector of the representation
    $L(\omega_i)$. It is well-known that the highest weight vector of the fundamental representation $L(\omega_i)$,
    for $i=1,\dots,n-1$, is
    $$
    v_i=u_1\wedge u_2\wedge \dots\wedge u_i\,,.
    $$
    Moreover $\widehat{v}_n=u_1\wedge u_2\wedge \dots\wedge u_n$ is the highest weight vector of $\bigwedge^nL(\omega_1)$.

    A set of Chevalley generators for $\wdd$ is obtained as follows ($i =1,\dots,n-1$):
    \begin{align*}
    &
    e_0=\frac{\epsilon_n}2(F_{n+1,1}-2G_{n+2,1})t\,,
    &&
    e_{i}=\widetilde e_{i}
    \,,
    &&
    e_n=\widetilde{e}_{n}+\widetilde{e}_{n+1}\,,
    \\
    &
    h_0=-2F_{11}+2c\,,
    &&
    h_{i}=\widetilde{h}_i
    \,,
    &&
    h_n=\widetilde{h}_n+\widetilde{h}_{n+1}
    \,,
    \\
    &
    f_0=\epsilon_n^{-1}(2F_{1,n+1}-G_{1,n+2})t^{-1}\,,
    &&
    f_{i}=\widetilde{f}_i
    \,,
    &&
    f_n=\widetilde{f}_n+\widetilde{f}_{n+1}
    \,,
    \end{align*}
    where $\epsilon_n=1$ if $n$ is even and $\epsilon_n=\sqrt{-1}$ if $n$ is odd.
\begin{remark}
Notice that if $n$ is odd then $e_0$ is not real in $V^{(1)}$, even though the spectrum is invariant under complex conjugation. The appearence of a non-real matrix $\Lambda$ is unavoidable: in fact, as shown below, we have $V^{(n)}=L(\omega_n)_{-\frac14}$ , and $\Lambda$  cannot have a real maximal eigenvalue in this representation if $e_0$ is a real matrix in the representation $L(\omega_1)_0$.
    \end{remark}
Using the numbers $k_i$'s for $\widetilde{\mf g}$ given by Table \ref{table:ki}, we have that
$$V^{(i)}=\bigwedge^i L(\omega_1)_{\frac{c_i}{4}}\,,\quad i=1,\dots,n-1\,,$$
where $c_i$= 1 if $i$ is even and  $c_i=0$ if $i$ is odd.
We also define $U^{(n)}=\bigwedge^n L(\omega_1)_{c_n}$.
Moreover, we have that
$$
V^{(n)}=L(\omega_{n})_{d_n} \qquad\text{and}\qquad
V^{(n+1)}=L(\omega_{n+1})_{d_n}\,,
$$
where $d_n=\frac12$ if $n$ is even and $d_n=-\frac14$ if $n$ is odd.
In particular, $V^{(1)}=\mb C^{2n+2}$.
The matrix $R_1$ that realizes the isomorphism $R_1: V^{(1)} \to V^{(1)}_{\frac12}$ given in \eqref{isoR_i}  is
\begin{equation}\label{eq:R1D}
R_1= \sum_{k=1}^n \big( E_{kk}+ E_{k'k'} \big) + \frac12 E_{n+1,n+2}+ 2 E_{n+2,n+1}.
\end{equation}
Recall that $\Lambda=e_0+e_1+\dots+e_n$.
    Let us set
    \begin{equation}\label{app:psi_1_D}
    \psi^{(1)}=\sum_{k=1}^{n}\big( u_k +u_{k'} \big)
    +\frac{\epsilon_n+1}{2}u_{n+1} + (1-\epsilon_n) u_{n+2}
    \,.
    \end{equation}
    Then, it is easy to check that $\Lambda \psi^{(1)}=\psi^{(1)}$.
    By equation \eqref{20151003:eq1}, it follows that
    $\psi^{(1)}_{j}\in V^{(1)}$ is an eigenvector
    with eigenvalue $\gamma^j$, for every $j=0,\dots,n$.
    Similarly $R_1(\psi^{(1)}_{\frac12-j})\in V^{(1)} $ is an eigenvector
    with eigenvalue $\gamma^{j-\frac12}$ for every $j=0,\dots,n$..
    Hence, for the representation $V^{(1)}$, the eigenvalue $\lambda^{(1)}=1$ is a maximal eigenvalue
    with corresponding eigenvector $\psi^{(1)}$.
    Furthermore, by Lemma \ref{lem:pisa13gen}, it follows that
    \begin{equation}\label{eq:lambda_Dneqn}
    \lambda^{(i)}=\gamma^{-\frac{i-1}4}\sum_{j=0}^{i-1}\gamma^{\frac j2}
    =\frac{\sin\left(\frac{i\pi}{2n+2}\right)}{\sin\left(\frac{\pi}{2n+2}\right)}\,,
    \end{equation}
is a maximal eigenvalue of $\Lambda$ in the representation $V^{(i)}$, for every $i=1,\dots,n-1$, and the corresponding eigenvector is easily checked to be
$$
\psi^{(i)}=\phi_{-\frac{i-1}4}^{(1)}\wedge\phi_{-\frac{i-3}4}^{(1)}
\wedge\dots\wedge \phi_{\frac{i-3}4}^{(1)}\wedge\phi_{\frac{i-1}4}^{(1)}\in V^{(i)},$$
where
$$
\phi^{(1)}_k=\left\{\begin{array}{ll}\psi^{(1)}_k\,, & k \equiv 0,\frac{1}{4} \bmod{1}\,,\\
R_1(\psi^{(1)}_k)\, & k\equiv\frac12,\frac34\bmod{1}\,.
\end{array}\right.
$$
The same formulas for $i=n$ yield the maximal eigenvalue $\widehat{\lambda}^{(n)}$ and the corresponding eigenvector $\widehat{\psi}^{(n)}$ of $U^{(n)}$.
    Since $\Ht(j)=j-1$, for $j=1,\dots,n$ (in fact $u_{j+1}=f_{j}u_{j}$ for $j \leq n-1$), we can use
    Lemma \ref{lem:vadermonde}, with $F(j)=j$, for $1\leq j\leq i$, and the explicit expression for $\psi^{(1)}$ given by \eqref{app:psi_1_D}
    and for $R_1$ given by \eqref{eq:R1D} to get
    \begin{equation}\label{app:D_4}
    \psi^{(i)}= \kappa_i v_i+ y_i\,,
    \qquad i=1,\dots, n-1
    \,,
    \end{equation}
    where $\kappa_i \neq 0$ and $y_i$ is a combination of lower weight vectors.
    Similarly $\widehat{\psi}^{(n)}=\kappa_n \widehat{v}_n + y_n $, with $\kappa_n \neq 0$.

    We note that the spectrum of $\Lambda$ in the representations
    $\bigwedge^{i}L(\omega_1)_{c_i}$, for $i=1,\dots,n-1$, coincides with the spectrum
    of the element $\Lambda$ of the Kac-Moody algebra $A^{(1)}_{2n+1}$
    in the representations $\bigwedge^iL(\omega_1)_{\frac{1+(-1)^i}{4}}$.
    This spectrum is well-known
    and was also analyzed in our previous paper \cite{marava15}.
    Thus, if $i=1,\dots,n-1$, for any $0\leq \alpha < \frac{1}{4}$, the representation $ \bigwedge^{i}L(\omega_1)_{\frac{1+(-1)^i}{4}+\alpha}$
    has a unique eigenvalue with maximal real part, namely $\gamma^{\alpha}\mu^{(i)}$,
    where $\mu^{(i)}=\frac{\sin\left(\frac{i\pi}{2n+2}\right)}{\sin\left(\frac{\pi}{2n+2}\right)}$.
    Furthermore, the $\mu^{(i)}$'s satisfy the relation
    \begin{equation}\label{eq:mui}
    (\gamma^{\frac 14}+\gamma^{-\frac 14})\mu^{(i)}=\mu^{(i-1)}+\mu^{(i+1)} \,.
    \end{equation}
    While for $\alpha=\frac14$, the representation $ \bigwedge^{i}L(\omega_1)_{\frac{1+(-1)^i}{4}+\alpha}$
    has two eigenvalues with maximal real part, namely $\gamma^{\pm \alpha}\mu^{(i)}$.
    Let us now prove part i) of Theorem \ref{thm:genlamba} for $i=n$.
    Due to the decomposition \eqref{eq:embeddingDn_3} and the fact that $V^{(n+1)}_{\frac 12}\cong V^{(n)}$
    (using the morphism $R_{n+1}$), we have that
    \begin{equation}\label{eq:spinspindecomp}
    V^{(n)}_{-\frac14 + \alpha} \otimes V^{(n)}_{\frac14 +\alpha}
    \cong U^{(n)}_{\alpha}\oplus V^{(n-2)}_\alpha \oplus \dots
    \,.
    \end{equation}
    For $0\leq\alpha <\frac14$ the representation given in equation \eqref{eq:spinspindecomp} has one eigenvalue
    with real maximal part $\gamma^{\alpha} \mu^{(n)}$,
    where
    $\mu^{(n)}=\frac{\sin\left(\frac{n\pi}{2n+2}\right)}{\sin\left(\frac{\pi}{2n+2}\right)}$.
    Therefore, it follows that in the representation $V^{(n)}_{\pm \frac14 + \alpha}$ the element $\Lambda$
    has one maximal eigenvalue of the form $\gamma^{\alpha} y_{\pm \frac14}$,
    with $y_{\frac14}+y_{-\frac14}=\mu^{(n)}$.

    If $ \alpha=\frac14$, then in the representation$V^{(n)}_{-\frac14 + \alpha} \otimes V^{(n)}_{\frac14 +\alpha} $ the element
    $\Lambda$ has two complex conjugates eigenvalues
    with maximal real part, namely $\gamma^{\pm\alpha}\mu^{(n)}$.
    Since the spectra of $\Lambda$ in the representations $V^{(n)}$ and $V^{(n)}_{\frac12}$
    are invariant under complex conjugation \footnote{The matrix $\Lambda$ is in fact real in the spin representations with our choice of $\e_n$, $k_{n}$, $k_{n+1}$.},
    we can conclude that in one of the two representations $V^{(n)}$ and $V^{(n)}_{\frac12}$,
    $\Lambda$ has a maximal eigenvalue, while in the other representation, it has two complex conjugated eigenvalues with
    maximal real part.
    Therefore we can find a suitable normalization of $\Lambda$ such that it has a maximal eigenvalue in the representation $V^{(n)}$.
    In fact, if $\Lambda$ does not have a maximal eigenvalue in the representation $V^{(n)}$, then it acquires it after sending
    $\epsilon_n\to-\epsilon_n$ in the choice of Chevalley generators for $D_{n+1}^{(2)}$. However, this change in the choice of the
    Chevalley generators does not alter the spectrum of $\Lambda$ in the representations $V^{(i)}$, for $i=1,\dots,n-1$.
    Hence, we have that $y_{-\frac 14}\gamma^{\frac14} =\lambda^{(n)}$ is a real positive number and
    $\lambda^{(n)} = \frac{\mu^{(n)}}{\gamma^{\frac14}+\gamma^{-\frac14}}$. Using the identities (\ref{eq:mui}) together with
    $\mu^{(n)}=\mu^{(n+2)}$, we arrive at the identity
    \begin{equation}\label{eq:rellambdanDn}
     \lambda^{(n)}=\frac{\lambda^{(n-1)}}{\gamma^{\frac 12}+\gamma^{-\frac 12}}
    \end{equation}
    Note that, in this particular case, equation \eqref{eq:genlambdarelations} reads
    \begin{align}
    \begin{split}\label{eq:Dn+1lambdarel}
    &(\gamma^{-\frac14}+\gamma^{\frac{1}{4}})\lambda^{(i)}=\lambda^{(i-1)}+\lambda^{(i+1)}\,,
    \qquad i =1,\dots,n -2 \,,
    \\
    &(\gamma^{-\frac14}+\gamma^{\frac{1}{4}})\lambda^{(n-1)}=\lambda^{(n-2)} +
    (\gamma^{-\frac14}+\gamma^{\frac{1}{4}})\lambda^{(n)}  \, ,
    \\
    &(\gamma^{-\frac12}+\gamma^{\frac12})\lambda^{(n)}=\lambda^{(n-1)}
    \,,
    \end{split}
    \end{align}
    where we set $\lambda^{(0)}=0$.
    Using equations \eqref{eq:lambda_Dneqn} and \eqref{eq:rellambdanDn} one can check directly that
    (\ref{eq:Dn+1lambdarel}) is satisfied thus proving part ii).

    Part iii) and iv) in the case $i \leq n-2$ can be proved in the same way as done for the analogue statement
    for the Lie algebra $B_n$ in the case $i \neq n$.
    The details of the proof are therefore omitted.

    Let us conclude the proof of part iii).
    For $i=n-1$ it follows from the fact that $V^{(n)}_{-\frac14} \otimes V^{(n)}_{\frac14}$ has a
    maximal eigenvalue as it was proved in the proof of part i).
    For the case $i=n$ we need to show that $V^{(n-1)}$
    has a maximal eigenvalue which also has already been proved in part i).

    Now we can conclude the proof of part iv). For $i=n-1$ it follows by noticing that $\bigwedge V^{(n-1)}_{\frac14}$ is embedded in
    $V^{(n-2)} \otimes U^{(n)}$, cf. \eqref{eq:embeddingDn_2}. Similarly, for $i=n$, it follows from the decomposition
    $\bigwedge V^{(n-1)}_{\frac12} \cong V^{(n-1)} \oplus V^{(n-5)} \oplus \dots$,
    cf. \eqref{eq:embeddingDn_4}.

    We are then left to show part v).
    The proof is analogous to the one for the $B_n$ case.
    Using the embeddings (\ref{eq:embeddingDn}) and the decompositions (\ref{eq:embeddingDn_3},\ref{eq:embeddingDn_4}),
    it is sufficient to show that the vector
    $\psi_{\otimes}^{(i)}\in\bigotimes_{j\in I}\bigotimes_{\ell=0}^{B_{ij}-1}V^{(j)}_{\frac{B_{ij}-1-2\ell}{2r}} $ has
    a non-trivial component lying in the highest weight subrepresentation generated by the highest weight vector
    $ \bigotimes_{j\in I}\bigotimes_{\ell=0}^{B_{ij}-1}v_j$.
    In turn, this follows from the fact that $\psi^{(i)}$ has a non-trivial component in the highest weight vector of $V^{(i)}$.
    For $i \neq n$, this follows from the decomposition (\ref{app:D_4}).
    The case $i=n$ is proved using the decomposition (\ref{eq:spinspindecomp}) for $\alpha=0$ as follows:
    Note that, from (\ref{app:D_4}), it follows that the vector $\widehat{\psi}^{(n)} \in U^{(n)}$
    has a non-trivial component in the highest weight vector $\widehat{v}_n=u_1 \wedge \dots \wedge u_n$ which is the image of $v_{n} \otimes v_{n}$
    under the decomposition \eqref{eq:spinspindecomp}. We have seen in the proof of part i) that  $\widehat{\psi}^{(n)}$ is
    identified under the decomposition \eqref{eq:spinspindecomp} with the vector $\psi^{(n)}_{-\frac14 } \otimes \psi^{(n)}_{\frac14}$. Hence
    $\psi^{(n)}_{-\frac14 } $ must have a non-trivial component in the highest weight vector $v_n$ and, since $h$ acts diagonally,
    $\psi^{(n)}$ has a non-trivial component in the vector $v_n$.

    \subsection{Proof of Theorem \ref{thm:genlamba} for $\mf g$ of type $G_2$}\label{sec:proof_G2}
    The dual Coxeter number of $\mf g^{(1)}$ is $h^\vee=4$
    thus we set $\gamma=e^{\frac{\pi i}2}$.
    Recall that in this case we have ${}^L\mf g^{(1)}=D_{4}^{(3)}$.

    Using the same notation of Section \ref{sec:proof_Cn} with $n=3$, the
    Chevalley generators of $\widetilde{\mf g}=D_4$ are:
    \begin{align*}
    &\widetilde{f}_i=F_{i+1,i}\,,
    &&\widetilde{h}_i=F_{ii}-F_{i+1,i+1}\,,
    &&\widetilde{e}_i=F_{i,i+1}\,,
    &i=1,2,3\,,
    \\
    &\widetilde{f_4}=2G_{53}\,,
    &&\widetilde{h}_4=F_{33}+F_{44}\,,
    &&\widetilde{e}_4=\frac12G_{35}
    \,.
    \end{align*}
    It is well-known that the representation $L(\omega_1)$ is given by the natural
    action of $\wg$ on $L(\omega_1)=\mb C^{8}$.
    Moreover we have that
    $$
    L(\omega_2) \cong \bigwedge^2L(\omega_1)\,,
    \qquad
    L(\omega_3)\cong L(\omega_1)^{\sigma^2}\,,
    \qquad
    L(\omega_4)\cong L(\omega_1)^{\sigma} \,,
    $$
    where $\sigma$ is the Dynkin diagram automorphism of $D_4$ defined in Table \ref{table:graphs}.
    A set of Chevalley generators of $D_4^{(3)}$ is:
    \begin{align*}
    &e_0= \kappa \left(e^{-i \frac{\pi}{3}} F_{41}+ 2G_{62}+2 e^{i\frac\pi3}G_{51}\right)t\,,
    &&e_1=\widetilde e_1+\widetilde e_3+\widetilde e_4\,,
    &&e_2=\widetilde{e}_2\,,
    \\
    &h_0=-2F_{11}-F_{22}-F_{33}+3c\,,
    &&h_1=\widetilde h_1+\widetilde h_3+\widetilde h_4\,,
    &&h_2=\widetilde{h}_2\,,
    \\
    &f_0= \frac{1}{\kappa}\left(e^{i \frac{\pi}{3}} F_{14}+ \frac12 G_{26}+\frac12 e^{-i\frac\pi3}G_{15}\right)t^{-1}\,,
    &&f_1=\widetilde f_1+\widetilde f_3+\widetilde f_4\,,
    &&f_2=\widetilde{f}_2\,.
    \end{align*}
    with $\kappa=\frac{1}{3+2 \sqrt{3}}$.
    Using the definitions of the $k_i$'s in Table \ref{table:ki} we have
    \begin{equation*}
     V^{(1)}=L(\omega_1)_0
     \, ,\qquad
     V^{(2)}=L(\omega_2)_{\frac12}\,.
    \end{equation*}
    The characteristic polynomial of the matrix $\Lambda$ in the representation $V^{(1)}$ is
    $p_1(x)=\left(x^4-4 \sqrt{3}+7\right) \left(x^4-1\right)$. Therefore the
    the matrix $\Lambda$ has maximal eigenvalue $\lambda^{(1)}=1$.
    The vector $\psi^{(1)}$ can be explicitly computed and its components in the standard basis of $\bb{C}^8$ are all non-zero.
    Hence, the proof of Theorem \ref{thm:genlamba} in this case
    follows the same lines as for the case of the non simply-laced Lie algebra of type $C_n$ in Section \ref{sec:proof_Cn}.
    We omit the details. However we plot in Figure \ref{fig:d43} the spectrum of $\Lambda$ in $V^{(1)}$ and $V^{(2)}$. Remarkably, even though the specturm in $V^{(1)}$
    has a $\bb{Z}/4\mb Z$ invariance, the spectrum of $V^{(2)}$ has a $\bb{Z}/12\mb Z$ invariance, as predicted by the theory. This reflects in the
    special form of the characteristic polynomial $p_1(x)$ of $\Lambda$ in $V^{(1)}$.

    \begin{figure}[htpb]
    \label{fig:d43}
      \centering
    \includegraphics[width=8cm]{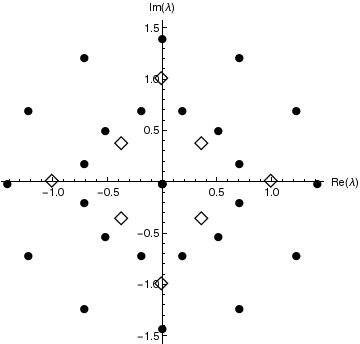}
    \caption{The spectrum of $\Lambda$ for the algebra $D_4^{(3)}$ in the represenations $V^{(1)}$ and $V^{(2)}$.}
    \end{figure}
    %

    \subsection{Proof of Theorem \ref{thm:genlamba} for $\mf g$ of type $F_4$}\label{sec:proof_F4}
    Recall that in this case we have ${}^L\mf g^{(1)}=E_{6}^{(2)}$.

    In order to study the algebra $E_6^{(2)}$ we use the realization of $E_6$ proposed in \cite{howlett01} by $27 \times 27$ matrices, coinciding
    with the representation $L(\omega_1)$.
    It is well-known from representation theory \cite{FH91}, that the representation $L(\omega_6)$ is the dual representation of $L(\omega_1)$.
    Moreover, we have the following isomorphisms of representations of the Lie algebra $E_6$:
    $$
    L(\omega_2)=\bigwedge^2 L(\omega_1)\,,
    \qquad
    L(\omega_3)=\bigwedge^3 L(\omega_1) \cong \bigwedge^3L(\omega_6)\,,
    \qquad
    L(\omega_5)=\bigwedge^2L(\omega_6)
    \,.
    $$
    Finally, $L(\omega_{4})$ is the adjoint representation and using standard tools in representation theory we compute
    \begin{equation}\label{eq:e6embedding}
    \bigwedge^2 L(\omega_4)\cong L(\omega_3) \oplus L(\omega_4)
    \,.
    \end{equation}
    Using the numbers $k_i$, $i=1,\dots,6$, for the Lie algebra $E_6$ given by Table \ref{table:ki}, we have
    \begin{equation}\label{20151010:repE6}
    \begin{array}{ccc}
    V^{(1)}=L(\omega_{1})\,,
    &
    V^{(2)}=L(\omega_{2})_{\frac12}\,,
    &
    V^{(3)}=L(\omega_3)\,,
    \\
    V^{(4)}=L(\omega_4)_{\frac14}\,,
    &
    V^{(5)}=L(\omega_{5})_{\frac12}\,,
    &
    V^{(6)}=L(\omega_{6})\,.
    \end{array}
    \end{equation}
    Following \cite{kac90} we can compute a set of Chevalley generators of the Kac-Moody algebra $E_6^{(2)}$ and then
    the characteristic polynomial of $\Lambda$ in the representations listed in \eqref{20151010:repE6}.
    The characteristic polynomials $p_1(x)$ and $p_4(x)$ of $\Lambda$ in the representations $V^{(1)}$ and $V^{(4)}$ are
    $$
    p_1(x)= -x^{27}+168 x^{18}+636 x^9+8
    $$
    and
    $$
    p_4(x)= (x^{72}-2709288 x^{54}-7822776528 x^{36}+250804880064 x^{18}-3673320192 )x^6\,.
    $$
    Note that all the characteristic polynomials of $\Lambda$ in the remaining representations in \eqref{20151010:repE6} can be computed from $p_1(x)$.
    Therefore the eigenvalues in all the representations in \eqref{20151010:repE6} can be computed finding the
    roots of a third and fourth order polynomial.
    We let the reader verify that in any representation $V^{(i)}$ listed in \eqref{20151010:repE6}
    there exists a maximal eigenvalue $\lambda^{(i)}$, and that the $\lambda^{(i)}$'s
    satisfy the relations (\ref{eq:genlambdarelations}). Moreover, in the weight vectors basis proposed in \cite{howlett01},
    the coefficients of the eigenvector $\psi^{(1)}$ are all strictly positive.
    Using these facts and the decomposition (\ref{eq:e6embedding}), the proof of Theorem \ref{thm:genlamba}
    follows using the same method outlined for the cases of the non simply-laced Lie algebras $B_n$ and $C_n$ in Sections
    \ref{sec:proof_Bn} and \ref{sec:proof_Cn} respectively.
    We omit the details.


\section{The \texorpdfstring{$Q$}{Q}-system (\texorpdfstring{$\mf g$}{g}-Bethe Ansatz)}\label{sec:qsystem}
Following the construction in the simply-laced case \cite{marava15}, 
in order to obtain the $Q$-system (Bethe Ansatz) for the algebra $\mf g$ we consider the monodromy around the Fuchsian singularity $x=0$
of solutions to the linear differential equation \eqref{20141021:eq1} obtained applying the connection \eqref{20141020:eq1} to an
evaluation representation of ${}^L\mf g^{(1)}$. 
More precisely, for $i\in I$ we define a function $Q^{(i)}(E)$ as the coefficient of the most singular part in the expansion around
$x=0$ of the fundamental solution $\Psi^{(i)}$ in the representation $V^{(i)}$. In the case $A_1^{(1)}$, the function $Q^{(1)}$ thus obtained corresponds
to the spectral determinant of a Schr\"odinger operator. Quite naturally, the functions $Q^{(i)}$, $i\in I$, will be called below
\emph{generalized spectral determinants}.

We restrict in this section to the case when $M h^\vee\in\bb{Z}_{+}$. From the general theory of Fuchsian singularities of linear ODEs,
the singular behaviour at $x=0$  depends on the spectrum of the element $\ell\in \wh_0$.
Therefore, we consider a dominant weight $\omega\in P^+$ of $\wg$, the related highest weight representation $L(\omega)$,
and the set   $P_\omega\subset P$ of weights appearing in the weight space decomposition of $L(\omega)$. Then,
we describe some properties of the spectrum of the element $\ell\in \wh_0\subset \wh$ when acting on  $L(\omega)$. The eigenvalues of
$\ell$ are of the form $\lambda(\ell)$ with $\lambda\in P_\omega$ a weight of $L(\omega)$, and a generic choice of $\ell\in\wh_0$ leads
to non resonant eigenvalues, namely $\lambda(\ell)-\lambda'(\ell)\notin\bb Z$ if $\lambda\neq\lambda'$. We require $\ell$ to be generic
in the above sense as well as regular with respect to the set of roots $R$ of $\wg$, namely to induce a decomposition of the form
$R=R^+_\ell\cup R^-_\ell$, where
$$
R^+_\ell=\lbrace \alpha\in R\mid \Re \alpha (\ell)>0 \rbrace\,,
\qquad
R^-_\ell= \lbrace \alpha\in R\mid\Re\alpha(l)<0 \rbrace
\,.
$$
Associated  to such a $\ell\in\wh_0$ there are a Weyl chamber $\mf{W}_\ell$, an element $w_\ell$ of the Weyl group of $\wg$, which maps  the fundamental
Weyl chamber relative to $\Delta\subset R$ into $\mf{W}_\ell$, and a set of simple roots
$\Delta_\ell=\lbrace w_\ell(\alpha_i)\mid \alpha_i\in\Delta, i\in \widetilde{I} \rbrace$. 
\begin{remark}\label{rem:weyl}
For any $\ell\in\wh$ the outer automorphism $\sigma$ acts on the Weyl group of $\wg$ by conjugation:
$w_{\sigma(\ell)}=\sigma w_\ell \sigma^{-1}$.
Since by definition  $\ell\in\wh_0$ is fixed by $\sigma$, it follows that the element $w_\ell$ obtained above
from a generic $\ell\in \wh_0$ commutes with $\sigma$, and therefore it belongs to the Weyl group of $\wg_0$. Indeed,
the latter can be charaterized as the subgroup of elements of the Weyl group of $\wg$ commuting with $\sigma$.
\end{remark}
In \cite{marava15} we proved the following result 
\begin{proposition}\label{prop:maximalweight}
Let $\ell \in \wh_0$ be a generic element, and let $w_\ell$ be the associated element of the Weyl group.
Then the weight $w_\ell(\omega)\in P_\omega$, $\omega\in P^+$, has multiplicity one in $L(\omega)$ and $\Re w_\ell(\lambda)(\ell) > \Re \lambda(\ell)$
for any weight $\lambda \in P_{\omega}$, $\lambda\neq w_\ell(\omega)$.
In the case of a fundamental weight $\omega=\omega_i$, $i\in I$, of $\wg$,
the weight $ w_\ell(\omega_i-\alpha_i)\in P_{\omega_i}$
has multiplicity one and
$\Re w_\ell(\omega_i)(\ell) > \Re w_\ell(\omega_i-\alpha_i)(\ell) > \Re \lambda(\ell),$
for any $\lambda \in P_{\omega_i}$, $\lambda\neq w_\ell(\omega_i),w_\ell(\omega_i-\alpha_i)$.
\end{proposition}
Let  now $\chi^{(i)}$, $\varphi^{(i)}\in L(\omega_i)$ be weight vectors corresponding to the weights $w_\ell(\omega_i)$ and $w_\ell(\omega_i-\alpha_i)$ respectively. Then, the elements
$$\chi^{(i)}\wedge \varphi^{(i)}\in \bigwedge^2 L(\omega_i), \qquad \otimes_{j\in I}\otimes_{\ell=0}^{B_{ij}-1}\chi^{(j)}\in \bigotimes_{j\in I}\bigotimes_{\ell=0}^{B_{ij}-1}L(\omega_{j})$$
are weight vectors with the same weight $\eta_i=\sum_{j\in\widetilde{I}}\widetilde{B}_{ij}\omega_j\in P^+$. We normalize $\chi^{(i)}$, $\varphi^{(i)}$, $i\in I$, in such a way that they are identified through the morphism \eqref{eq:laverami}:
\begin{equation}
m_i(\chi^{(i)}\wedge \varphi^{(i)})=\otimes_{j\in I}\otimes_{\ell=0}^{B_{ij}-1}\chi^{(j)}.
\end{equation}
As a consequence of the above results, and since $M h^\vee\in \bb{Z}_+$, from the general theory of Fuchsian singularities of linear ODEs it follows that for any evaluation representation $V^{(i)}_k$, $i\in I$ and $k\in\bb C$, there exist normalized solutions $\chi^{(i)}_k(x,E)$ and $\varphi^{(i)}_k(x,E)$ to equation \eqref{20141021:eq1}, which are eigenvectors of the monodromy matrix and have the most singular behavior at $x=0$, with asymptotic expansion:
\begin{align*}
\chi^{(i)}_k(x,E)&= x^{-w_\ell(\omega_i)(\ell)}(\chi^{(i)}+O(x)),\\
\varphi^{(i)}_k(x,E)&=x^{-w_\ell(\omega_i-\alpha_i)(\ell)}(\varphi^{(i)}+O(x)).
\end{align*}
In addition, $\chi^{(i)}_k(x,E)$ and $\varphi^{(i)}_k(x,E)$ are entire functions of $E$. Due to \eqref{20150108:eq8}, and as a result of a comparison of the asymptotic behaviors we obtain
\begin{align}
\begin{split}
& \omega^{-k h} \chi_{k'}^{(i)} (\omega^k x,\Omega^k E)
= \omega^{-k w_\ell(\omega_i)(\ell+h)} \chi^{(i)}(x,E)_{k+k'} \\ \label{eq:chik}
& \omega^{-k h} \phi_{k'}^{(i)}(\omega^k x,\Omega^k E)
= \omega^{-k w_\ell(\omega_i-\alpha_i)(\ell+h)} \phi^{(i)}(x,E)_{k+k'}
\,,
\end{split}
\end{align}
for every $k,k'\in\mb C$.
Since for generic $\ell\in\widetilde{\mf h}_0$,
the eigenvalues $-w_\ell(\omega_i)(\ell)$, and $-w_\ell(\omega_i-\alpha_i)(\ell)$
are non-resonant, we can uniquely define two functions $Q^{(i)}(E;\ell)$ and $\widetilde{Q}^{(i)}(E;\ell)$
as the coefficients of the following expansion in the invariant subspaces of the monodromy matrix,
\begin{equation}\label{psigrande}
\Psi^{(i)}(x,E,\ell)
=Q^{(i)}(E;\ell)\chi^{(i)}(x,E)+\widetilde{Q}^{(i)}(E;\ell) \phi^{(i)}(x,E)+v^{(i)}(x,E)
\,,
\end{equation}
where, $v^{(i)}(x,E)$ belongs to an invariant subspace of lower weights vectors.
We call $Q^{(i)}(E;\ell)$ and $\widetilde Q^{(i)}(E;\ell)$ the generalized
spectral determinants of the equation \eqref{20141020:eq1}.
The $\Psi$-system \eqref{eq:031201} implies the following non-trivial functional
relations among the generalized spectral determinants.
\begin{proposition}\label{thm:QQtilde}
Let $\ell\in\wh_0$ be  generic. Then, the spectral determinants $Q^{(i)}(E;\ell)$ and $\widetilde{Q}^{(i)}(E;\ell)$  are entire functions of  $E$ and satisfy the following  $Q\widetilde{Q}$-system:
\begin{align}
\begin{split}\label{QQsystem}
\prod_{j\in I}\prod_{\ell=0}^{B_{ij}-1}Q^{(j)}(\Omega^{\frac{B_{ij}-1-2\ell}{2r}}E)
&=\omega^{\frac{D_i\theta_i}{2}}Q^{(i)}(\Omega^{\frac{D_i}{2}}E)\widetilde{Q}^{(i)}(\Omega^{\frac{-D_i}{2}}E)\\
&-\omega^{-\frac{D_i\theta_i}{2}}Q^{(i)}(\Omega^{-\frac{D_i}{2}}E)\widetilde{Q}^{(i)}(\Omega^{\frac{D_i}{2}}E)
\,,
\end{split}
\end{align}
where $\theta_i=w_\ell(\alpha_i)(\ell+h)$.
\end{proposition}
\begin{proof}
Let $\ell$ be generic, so that the solutions $\chi^{(i)}$ and $\phi^{(i)}$ as defined above exist. Since equation \eqref{20141021:eq1} depends linearly on $E$, then we have that the functions $\Psi^{(i)}$,  $\chi^{(i)}$ and $\phi^{(i)}$, are entire with respect to the parameter $E$,  from which it follows that also $Q^{(i)}$ and $\widetilde{Q}^{(i)}$ are entire functions of $E$. In order to get the $Q\widetilde{Q}$-system, we note that the expansion \eqref{psigrande} implies the following one:
\begin{align*}
\psi^{(i)}_k(x,E)&=\omega^{-kH}\psi^{(i)}(\omega^kx,\Omega^kE)\\
&=Q^{(i)}(\Omega^kE)\omega^{-kH}\chi^{(i)}(\omega^kx,\Omega^kE)+\widetilde{Q}^{(i)}(\Omega^kE)\omega^{-kH}\varphi^{(i)}(\omega^kx,\Omega^kE)\!+\!\dots\\
&=Q^{(i)}(\Omega^kE)\omega^{-k\gamma_i}\chi^{(i)}_k(x,E)+\widetilde{Q}^{(i)}(\Omega^kE)\omega^{-k\delta_i}\varphi^{(i)}_k(x,E)\!+\!\dots ,
\end{align*}
where $\gamma_i=w_\ell(\omega_i)(\ell+h)$, $\delta_i=w_\ell(\omega_i-\alpha_i)(\ell+h)$ and the remaining unspecified terms are of lower weight. Substituting the above expansion into the $\Psi$-system \eqref{eq:031201}, we thus obtain:
\begin{gather*}
m_i\left(R_i\left(Q^{(i)}(\Omega^{-\frac{D_i}{2}}E)\omega^{\frac{D_i\gamma_i}{2}}\chi^{(i)}_{-\frac{D_i}{2}}(x,E)\right)
\wedge \widetilde{Q}^{(i)}(\Omega^\frac{D_i}{2}E)\omega^{-\frac{D_i\delta_i}{2}}\varphi^{(i)}_\frac{D_i}{2}(x,E)\right)\\
+m_i\left(R_i\left(\widetilde{Q}^{(i)}(\Omega^{-\frac{D_i}{2}}E)\omega^{\frac{D_i\delta_i}{2}}
\varphi^{(i)}_{-\frac{D_i}{2}}(x,E)\right)\wedge Q^{(i)}(\Omega^\frac{D_i}{2}E)\omega^{-\frac{D_i\gamma_i}{2}}\chi^{(i)}_\frac{D_i}{2}(x,E)\right)\\
=\bigotimes_{j\in I}\bigotimes_{\ell=0}^{B_{ij}-1} Q^{(j)}\left(\Omega^{\frac{B_{ij}-1-2\ell}{2r}}E\right)
\chi^{(j)}_{\frac{B_{ij}-1-2\ell}{2r}}(x,E) \omega^{-\gamma_j\frac{B_{ij}-1-2\ell}{2r}}+\dots\, .
\end{gather*}
Collecting on both sides above the coefficients of the highest weight vector, and introducing the quantities $\theta_i=\gamma_i-\delta_i=w_\ell(\alpha_i)(\ell+h)$, one arrives at the $Q\widetilde{Q}$-system \eqref{QQsystem}.
\end{proof}
In the simply-laced case the $Q$-system follows from the $Q\widetilde{Q}$-system after a one-line calculation. In the present case, some more work is needed.
\begin{lemma}\label{20151014:lem1}
For $i\in I$, we have $D_i\theta_i=\sum_{j\in I}\overline{C}_{ij}\beta_j$, with $\beta_j=w_\ell(\omega_j)(\ell+h)$.
\end{lemma}
\begin{proof}
Recall that from Proposition \ref{thm:QQtilde} that we defined $\theta_i=w_\ell(\alpha_i)(\ell+h)$, for every $i\in I$.
For every $i\in \widetilde{I}$, we have that $\alpha_i=\sum_{j\in \widetilde{I}}\widetilde{C}_{ji}\omega_j$.
Moreover, by the construction \eqref{chevg0} of the Chevalley generators of ${}^L\mf g^{(1)}$, it follows that $\omega_{j}(a)=\omega_{\sigma(j)}(a)$ for every $a\in \widetilde{\mf h}_0$ and $j\in\widetilde{I}$.
Now assume that $D_i\theta_i=\sum_{j\in I}\overline{C}_{ij}\beta_j$, for $i\in I$ and certain $\beta_j$. Then, for $i\in I$, we have the following identities:
\begin{align*}
\beta_j &=\sum_{i\in I}\left(\overline{C}^{-1}\right)_{ji} D_i \theta_i=\sum_{i\in I}\left(\overline{C}^{-1}\right)_{ji} D_i\, w_\ell(\alpha_i)(\ell+h)\\
&=\sum_{i\in I}\left(\overline{C}^{-1}\right)_{ji} D_i \,w_\ell\Big(\sum_{k\in \widetilde{I}}\widetilde{C}_{ki}\omega_k(\ell+h)\Big)\\
&=\sum_{i\in I}\left(\overline{C}^{-1}\right)_{ji} D_i \,w_\ell\Big(\sum_{k\in I}\sum_{\ell=0}^{\langle k\rangle-1}\widetilde{C}_{\sigma^\ell(k)i}\,\omega_{\sigma^\ell(k)}(\ell+h)\Big)\\
&=\sum_{i\in I}\left(\overline{C}^{-1}\right)_{ji} D_i \,w_\ell\Big(\sum_{k\in I}\sum_{\ell=0}^{\langle k\rangle-1}\widetilde{C}_{i\sigma^\ell(k)}\,\omega_{k}(\ell+h)\Big)\\
&=w_\ell\left(\sum_{i\in I} \left(\overline{C}^{-1}\right)_{ji}D_i\sum_{k\in I} C_{ik}\omega_k(\ell+h)\right)\\
&=w_\ell\left(\sum_{i\in I} \left(\overline{C}^{-1}\right)_{ji}\sum_{k\in I} \overline{C}_{ik}\,\omega_k(\ell+h)\right)=w_\ell(\omega_j)(\ell+h).
\end{align*}
The lemma is proved.
\end{proof}
We are now in the position to state our fundamental result:
\begin{theorem}\label{thm:betheansatz} 
Let $E^\ast$ be a zero of $Q^{(i)}$, $i\in I$, such that
\begin{equation}\label{20151014:nonzero}
\prod_{j\in I}\prod_{\ell=0}^{B_{ij}-1}Q^{(j)}(\Omega^{\frac{B_{ij}-1-2\ell}{2r}-\frac{D_i}{2}}E^\ast)\neq 0.
\end{equation}
Then the following identity, known as $Q$-system (or $\mf g$-Bethe Ansatz), holds:
\begin{equation}\label{eq:betheansatz}
\prod_{j\in I}\Omega^{\overline{C}_{ij}\overline{\beta}_j}\frac{Q^{(j)}(\Omega^{\frac{\overline{C}_{ij}}{2}}E^\ast)}{Q^{(j)}(\Omega^{-\frac{\overline{C}_{ij}}{2}}E^\ast)}=-1,
\end{equation}
where $\overline{\beta}_j=\frac{1}{Mh^\vee}w_\ell(\omega_j)(\ell+h).$
\end{theorem}
\begin{proof}
Let $E^\ast$ be a zero of $Q^{(i)}$. Evaluating the $Q\widetilde{Q}$-system \eqref{QQsystem} at $E=\Omega^\frac{D_i}{2}E^\ast$ we get
\begin{equation}\label{eq:160715-2}
\omega^{\frac{D_i\theta_i}{2}}\widetilde{Q}^{(i)}(E^\ast)Q^{(i)}(\Omega^{D_i}E^\ast)=\prod_{j\in I}\prod_{\ell=0}^{B_{ij}-1}Q^{(j)}(\Omega^{\frac{B_{ij}-1-2\ell}{2r}+\frac{D_i}{2}}E^\ast),
\end{equation}
while evaluating at $E=\Omega^{-\frac{D_i}{2}}E^\ast$ we obtain
\begin{equation}\label{eq:160715-3}
-\omega^{-\frac{D_i\theta_i}{2}}Q^{(i)}(\Omega^{-D_i}E^\ast)\widetilde{Q}^{(i)}(E^\ast)=\prod_{j\in I}\prod_{\ell=0}^{B_{ij}-1}Q^{(j)}(\Omega^{\frac{B_{ij}-1-2\ell}{2r}-\frac{D_i}{2}}E^\ast).
\end{equation}
Due to the assumption \eqref{20151014:nonzero} we can get rid of the function $\widetilde{Q}^{(i)}$, obtaining the identity
\begin{equation}\label{eq:160715-4}
\omega^{D_i\theta_i}\frac{Q^{(i)}(\Omega^{D_i}E^\ast)}{Q^{(i)}(\Omega^{-D_i}E^\ast)}\prod_{j\in I}\prod_{\ell=0}^{B_{ij}-1}\frac{Q^{(j)}(\Omega^{\frac{B_{ij}-1-2\ell}{2r}-\frac{D_i}{2}}E^\ast)}{Q^{(j)}(\Omega^{\frac{B_{ij}-1-2\ell}{2r}+\frac{D_i}{2}}E^\ast)}=-1.
\end{equation}
To deduce the $Q$-system from the above, we consider three cases. First of all, recall that if $i\neq j$ then by definition $C_{ij}=-B_{ij}$. Therefore, if $i\neq j$ and $B_{ij}=1$, then we have $\overline{C}_{ij}=D_i C_{ij}=-D_i$, and
$$\prod_{\ell=0}^{B_{ij}-1}\frac{Q^{(j)}(\Omega^{\frac{B_{ij}-1-2\ell}{2r}-\frac{D_i}{2}}E^\ast)}{Q^{(j)}(\Omega^{\frac{B_{ij}-1-2\ell}{2r}+\frac{D_i}{2}}E^\ast)}=\frac{Q^{(j)}(\Omega^{\frac{\overline{C}_{ij}}{2}}E^\ast)}{Q^{(j)}(\Omega^{-\frac{\overline{C}_{ij}}{2}}E^\ast)}.$$
On the other hand, if $i\neq j$ and $B_{ij}>1$, then $D_i=\frac{1}{r}$ and $B_{ij}=r$, from which it follows that $\overline{C}_{ij}=D_i C_{ij}=-D_iB_{ij}=-1$, and so
\begin{gather*}
\prod_{\ell=0}^{B_{ij}-1}\frac{Q^{(j)}(\Omega^{\frac{B_{ij}-1-2\ell}{2r}-\frac{D_i}{2}}E^\ast)}{Q^{(j)}(\Omega^{\frac{B_{ij}-1-2\ell}{2r}+\frac{D_i}{2}}E^\ast)}=\frac{\prod_{\ell=0}^{r-1}Q^{(j)}(\Omega^{\frac{r-2-2\ell}{2r}}E^\ast)}{\prod_{\ell=0}^{r-1}Q^{(j)}(\Omega^{\frac{r-2\ell}{2r}}E^\ast)}\\
=\frac{\prod_{\ell=1}^{r}Q^{(j)}(\Omega^{\frac{r-2\ell}{2r}}E^\ast)}{\prod_{\ell=0}^{r-1}Q^{(j)}(\Omega^{\frac{r-2\ell}{2r}}E^\ast)}
= \frac{Q^{(j)}(\Omega^{-\frac{1}{2}}E^\ast)}{Q^{(j)}(\Omega^{\frac{1}{2}}E^\ast)}=  \frac{Q^{(j)}(\Omega^{\frac{\overline{C}_{ij}}{2}}E^\ast)}{Q^{(j)}(\Omega^{-\frac{\overline{C}_{ij}}{2}}E^\ast)}.
\end{gather*}
Finally, note that $\overline{C}_{ii}=D_i C_{ii}=2D_i$, so that
$$\frac{Q^{(i)}(\Omega^{D_i}E^\ast)}{Q^{(i)}(\Omega^{-D_i}E^\ast)}=\frac{Q^{(j)}(\Omega^{\frac{\overline{C}_{ii}}{2}}E^\ast)}{Q^{(j)}(\Omega^{-\frac{\overline{C}_{ii}}{2}}E^\ast)}.$$
As a consequence of these identities, and due to the relation $\omega=\Omega^\frac{1}{Mh^\vee}$, equation \eqref{eq:160715-4} can be written as
$$\Omega^\frac{{D_i\theta_i}}{Mh^\vee}\prod_{j\in I}\frac{Q^{(j)}(\Omega^{\frac{\overline{C}_{ij}}{2}}E^\ast)}{Q^{(j)}(\Omega^{-\frac{\overline{C}_{ij}}{2}}E^\ast)}=-1.
$$
The last step we need is to write the quantities $D_i\theta_i$, $i\in I$, in terms of the components of the symmetrized matrix $\overline{C}$. This is done in Lemma \ref{20151014:lem1}. The theorem is proved.
\end{proof}

\subsection{Action of the Weyl group on solutions to the Bethe Ansatz}\label{sub:weyl}

It is known that the Weyl group acts on the space of solutions to the Bethe Ansatz equations, see e.g. \cite{mukhin08}.
We now show how the Weyl group of $\wgz$ acts on the solutions to the Bethe Ansatz equation
(\ref{eq:betheansatz}) in our construction, and we prove that for generic $\ell \in \whz$ the action is free.
We identify the Weyl group of $\wgz$ as the subgroup of the Weyl group of $\wg$ whose elements are fixed by the action of
the outer automorphism $\sigma$; see Remark \ref{rem:weyl}.

The action is described as follows. 
Fixed an element $w$ of the Weyl group, we notice that in the representation $L(\omega_i)$ the weights
$w_{\ell}(w(\omega_i))$ and $w_{\ell}(w(\omega_i-\alpha_i))$ have multiplicity one, as follows from 
Proposition \ref{prop:maximalweight} and the fact the multiplicity of a weight is invariant under the Weyl group action.
If $\ell$ is generic in the sense of Proposition \ref{prop:maximalweight},
we can define the two solutions $\chi^{(i)}_w(x,E)$ and $\phi_w^{(i)}(x,E)$
which are eigenvectors of the monodromy matrix with eigenvalues $e^{i 2 \pi w_{\ell}(w(\omega_i))(\ell) }$ and
$e^{i 2 \pi w_{\ell}(w(\omega_i-\alpha_i))(\ell)} $ respectively. We can therefore 
decompose the subdominant solution $\Psi(x,E)$ in (new) invariant subspaces of the monodromy matrix as follows
\begin{equation}\label{wpsigrande}
\Psi^{(i)}(x,E,\ell)
=Q^{(i)}_w(E;\ell)\chi^{(i)}_w(x,E)+\widetilde{Q}^{(i)}_w(E;\ell) \phi^{(i)}(x,E)+v_w^{(i)}(x,E)
\,,
\end{equation}
where $v_w^{(i)}(x,E)$ belongs to the other invariant subspaces of the monodromy matrix.
Equation (\ref{wpsigrande}) defines two new $Q$-functions, $Q^{(i)}_w(E;\ell)$ and $\widetilde{Q}^{(i)}_w(E;\ell)$.
Following step by step all the reasoning that led to Theorem \ref{thm:betheansatz} above, we show that
the functions $Q^{(i)}_{w}$ satisfy the Bethe Ansatz equation
\begin{equation}\label{eq:wbetheansatz}
\prod_{j\in I}\Omega^{\overline{C}_{ij}\overline{w(\beta)}_j}
\frac{Q_\omega^{(j)}(\Omega^{\frac{\overline{C}_{ij}}{2}}E^\ast)}{Q_\omega^{(j)}(\Omega^{-\frac{\overline{C}_{ij}}{2}}E^\ast)}=-1,
\end{equation}
where $\overline{w(\beta)}_j=\frac{1}{Mh^\vee}w_\ell(w(\omega_j))(\ell+h).$

Finally, we note that the stabilizer of the tuple $(\omega_1,\dots,\omega_{n})$
consists of the identity element only. In fact, if an element in the Weyl group of $\wgz$ stabilizes $(\omega_1,\dots,\omega_{n})$
then, since it is invariant under conjugation by $\sigma$, it stabilizes the tuple $(\omega_1,\dots,\omega_{\tilde{n}})$, where $\tilde{n}$
is the rank of $\wg$; it is well known that the unique element of the Weyl group of $\wg$ stabilizing $(\omega_1,\dots,\omega_{\tilde{n}})$
is the identity \cite{hump90}.
We deduce that for a generic $\ell$ also the stabilizer of the Bethe Ansatz solution $(Q^{(1)}, \dots, Q^{(n)})$ is trivial.
Indeed, the equality $Q^{(i)}(E,\ell)=Q^{(i)}_w(E,\ell)$ cannot hold for any $i \in I$, since
$\overline{w(\beta)}_j \neq \overline{\beta}_j$
for at least one $j \in I$.

\section{Airy functions for twisted Kac-Moody algebras}\label{app:airy}
 We consider in more detail the special case of \eqref{20141125:eq1} with a linear potential $p(x,E)=x$ and with $\ell=0$.
 The element $\ell=0$ is clearly non-generic, however the $Q$-functions can be easily defined in this case. Following \cite{marava15},
 we fix an element $w$ in the Weyl group of $\wgz$ and we call $v_i^w$ the weight vector of  $L(\omega_i)$ with weight $w(\omega_i)$.
 Finally $Q^{(i)}(E)$ is just the coefficient with respect to $v_i^w$ of $\Psi^{(i)}(x,E)_{|x=0}$, where in the case of a linear potential
 we further have
 $ \Psi^{(i)}(0,E)=\Psi^{(i)}(-E,0)$.
For a non simply-laced Lie algebra $\mf g$, and given an evaluation representation of ${}^L\mf g^{(1)}$, we look for
 solutions of the equation
 \begin{equation}\label{eq:31ott01}
 \Psi'(x) + \big( e +x e_0 \big) \Psi(x)=0
 \,,
 \end{equation}
 in the Airy-like form
 \begin{equation}\label{eq:31ott02}
 \Psi(x)=\int_{ \mathit{c}} e^{-xs} \Phi(s) ds\,,
 \end{equation}
 where $\Phi(s)$ is an analytic function and $\mathit{c}$ is some path in the complex plane. We
 differentiate and integrate by parts to get
 \begin{equation}\label{oltrelultima}
 \int_{\mathit{c}} e^{-xs} \big( -s + e + e_0 \frac{d}{ds}\big) \Phi(s) ds
 + e^{-xs }e_0 \Phi(s) \big|_{\mathit{c}}=0
 \,,
 \end{equation}
 where the integral $\mathit{c}$ must be chosen so that the boundary term vanishes.
 If this is the case, the Airy solution \eqref{eq:31ott02} can be expressed in terms of the  solution to the simpler equation
 \begin{equation}\label{eq:3nov01}
 \big(- s + e + e_0 \frac{d}{ds}\big) \Phi(s)=0
 \end{equation}
 which we solve below for three important examples.
 Before tackling the analysis of the Airy solutions, we recall a well-known asymptotic formula.
 \begin{lemma}\label{lem:steepest}
Fix $M \in \bb{Z}, M \geq 2$ and a point $s_* >0$. Let $\mathit{c}$ be the curve passing trough
 $s_*$ such that $ \Re s^{M}=\Re s_*^{M}$
 for any point $s$ on the curve (the positive orientation is the one for which $\Im s^{M}$ is increasing). This contour is called a
 Stokes line for the action $s^{M}$. It is asymptotic to the rays of argument $ \pm \frac{\pi }{2 M}$
 and it is the boundary of the Stokes sector $ \Sigma=\{ s \in \bb{C}\mid \Re s^{M}\geq \Re s_*^{M} \}$.
Let moreover $f(s)$ be a function, analytic in $s$ for $s \in \Sigma $ and with the asymptotics
 $\lim_{s \to +\infty}\frac{f(s)}{s^K}=1$ for some $K \in \bb{C}$. Then
 \begin{equation}\label{eq:steepestdescent}
 \int_{ \mathit{c}} f(s) e^{-xs+\frac{x^M}{M}} ds \sim
 \big(\frac{2 \pi}{M-1}\big)^{\frac12} x^{\frac{2K+2-M}{2(M-1)}}e^{-\frac{M-1}{M}x^{\frac{M}{M-1}}} \mbox{ for } x\gg 0 \; .
 \end{equation}
 \end{lemma}
 \begin{proof}
 The statement follows from a simple steepest descent analysis. See \cite[Chapter 4.7]{peter06}.
 \end{proof}
 \subsection{\texorpdfstring{$A_{2n-1}^{(2)}$}{B\_n} in the standard representation}\label{sec:airy_A2n-1}
We study the $A_{2n-1}^{(2)}$-Airy function, $n\geq3$, in the standard representation $V^{(1)}=\mb C^{2n}$. Namely, we want to find
 a solution for equation \eqref{eq:3nov01} in the representation $V^{(1)}$ of the twisted Kac-Moody algebra
 $A_{2n-1}^{(2)}$.
 Using the explicit form of Chevalley generators $e_i$, $i=0,\dots,n$, which is provided in Section \ref{sec:proof_Bn},
 and denoting by $\Phi_i(s)$ the components of
 $\Phi(s)$ in the standard basis of $\mb{C}^{2n}$,
 equation \eqref{eq:3nov01} reads
 \begin{equation}\label{20151014:eq1}
 \left\{
 \begin{array}{ll}
 \Phi_{i+1}(s)=s\Phi_i(s)\,, & i=1,\dots,2n-2\,,
 \\
 \Phi'_1(s)+2\Phi_{2n}(s)=2s\Phi_{2n-1}(s)&
 \\
 \Phi'_2(s)=2s\Phi_{2n}(s)\,.&
 \end{array}
 \right.
 \end{equation}
 The general solution of \eqref{20151014:eq1} is
 \begin{align*}
 &
 \Phi_1(s)=\kappa s^{-\frac12}e^{\frac{s^{2n}}{2n}}\,,
 \qquad\Phi_i(s)=s^{i-1}\Phi_1(s)\,,
 \quad i=2,\dots,2n-1
 \,,
 \\
 &
 \Phi_{2n}(s)=\frac14\left(2s^{2n-1}+\frac{1}{s}\right)\Phi_1(s)
 \,,
 \end{align*}
 where $\kappa$ is an arbitrary complex number.
 We choose the path $\mathit{c}$ as in Lemma \ref{lem:steepest} with $M=2n$.
 It eventually lies on the rays of arguments $e^{\pm \frac{\pi i}{4n}}$, and it is clockwise oriented. On this path the exponential function is oscillatory,
 therefore fixed $\e >0$, for any $x$ such that $|\arg x| < \frac\pi{2} (1-\frac1{2n}) - \e $ the integral formula \eqref{eq:31ott02}
 and the boundary terms in \eqref{oltrelultima} vanish.
 By equation (\ref{eq:steepestdescent}) we get
 \begin{equation}
 \int_{ \mathit{c}} e^{-xs} s^j \Phi_1(s) ds \sim \sqrt{\frac{2 \pi}{2n-1}}
 x^{\frac{2j+1-2n}{4n-2}} e^{-\frac{2n-1}{2n}x^{\frac{2n}{2n-1}}}\,,
 \quad x\gg 0
 \,,
 \end{equation}
 and thus the Airy function coincides with the fundamental solution $\Psi^{(1)}$ - up to a multiplicative constant.
 The solutions $\Psi^{(1)}_k$, $k \in \mb{Z}$, $|k| <n$, are obtained by integrating \eqref{eq:31ott02} along the contour
 obtained rotating $\mathit{c}$ by $e^{\frac{ k \pi i }{2n}}$.
 
 \subsection{\texorpdfstring{$D^{(2)}_{n+1}$}{C\_n} in the standard representation}
\label{sec:airy_Dn+1}
We study the $D^{(2)}_{n+1}$-Airy function, $n\geq2$, in the standard representation. Namely, we want to find
 a solution for equation \eqref{eq:3nov01} in the representation $V^{(1)}=\mb C^{2n+1}$ of the twisted Kac-Moody algebra
 $D_{n+1}^{(2)}$.
 Using the explicit form of Chevalley generators $e_i$, $i=0,\dots,n$, which is provided in Section
 \ref{sec:proof_Cn}, and denoting by $\Phi_i(s)$ the components of
 $\Phi(s)$ in the standard basis of $\mb{C}^{2n+2}$,
 equation \eqref{eq:3nov01} reads
 \begin{equation}\label{20151014:eq2}
 \left\{
 \begin{array}{ll}
 \Phi_{i+1}(s)=s\Phi_i(s)\,, & i=1,\dots,n-1\,,
 \\
 &\text{ or }i=n+3,\dots,2n+1\,,
 \\
 2\Phi_{n+1}+\Phi_{n+2}(s)=2s\Phi_n(s)&
 \\
 \epsilon_n\Phi_1'(s)+\Phi_{n+3}(s)=2s\Phi_{n+1}(s)
 &
 \\
 -\epsilon_n\Phi_1'(s)+\Phi_{n+3}(s)=s\Phi_{n+2}(s)&
 \\
 2\Phi_{n+1}'(s)-\Phi_{n+2}'(s)=2\epsilon_n s\Phi_{2n+2}(s)
 \,.
 \end{array}
 \right.
 \end{equation}
 The general solution of \eqref{20151014:eq2} is
 \begin{align*}
 &\Phi_{i}(s)=s^{i-1}\Phi_1(s)\,, \qquad i=1,\dots,n\,,
 \\
 &\Phi_{n+1}(s)=\frac{s^n}{2}\Phi_1(s)+\frac{\epsilon_n}{2s}\Phi_1'(s)\,,
 \\
 &\Phi_{n+2}(s)=s^n\Phi_1(s)-\frac{\epsilon_n}{s}\Phi_1'(s)\,,
 \\
 &\Phi_{n+2+i}(s)=s^{n+i}\Phi_1(s)\,,\qquad i=1,\dots,n\,,
 \end{align*}
 where $\Phi_1(s)$ is a general solution to the Sturm-Liouville type equation
 \begin{equation}\label{eq:sturmDn}
 \left(\frac1s\Phi_1'(s)\right)'=s^{2n+1}\Phi_1(s)
 \,.
 \end{equation}
 Even though equation \eqref{eq:sturmDn} has a two dimensional space of solutions, these corresponds to a one dimensional space of solutions to equation (\ref{eq:31ott01}), as the integral  (\ref{eq:31ott02}) vanishes on a one dimensional subspace of solutions of (\ref{eq:sturmDn}). To prove this fact we proceed as follows.
 The function $\Pi(s)=s^{-\frac12}\Phi_1(s)$ satisfies
 \begin{equation}\label{eq:transformed}
 \Pi''(s)= \left(s^{2 n+2}+\frac{3}{4 s^2}\right) \Pi(s) \; .
 \end{equation}
 To study the integrals of the form (\ref{eq:31ott02}) we use the WKB analysis of equation \eqref{eq:transformed}.
 The (semiclassical) action of the potential $s^{2 n+2}+\frac{3}{4 s^2}$ is simply $ \frac{s^{n+2}}{n+2}$.
 We define therefore the Stokes line $\mathit{c}$ and the Stokes Sector $\Sigma$ as in Lemma \ref{lem:steepest} with $M=n+2$.
 From standard WKB theory \cite{fedoryuk93}, we know that there exist two linearly independent solutions of (\ref{eq:transformed}) which inside the Stokes sector
 have the following uniform asymptotics
 \begin{equation*}
 \Pi_{\pm} \sim s^{-\frac{n+1}{2}}e^{\pm \frac{s^{n+2}}{n+2} }
 \,.
 \end{equation*}
 Notice that the integral $ \int_{ \mathit{c}} s^k e^{-xs} \Pi_{\pm}(s) ds$ is well defined for any $k \in \bb{C}$
 if $x>0$ because $\Pi_{\pm}(s) $ are oscillatory on the Stokes line. Moreover, by the Cauchy Theorem
 we have that $ \int_{ \mathit{c}} s^k e^{-xs} \Pi_{-}(s) ds =0$
 if $x>0$, because $s^k e^{-xs} \Pi_{-}(s)$ is both analytic and exponentially small in the Stokes sector.
 After the above discussion and using the asymptotic expansion (\ref{eq:steepestdescent}) we obtain
 \begin{equation*}
 \int_{ \mathit{c}} s^k e^{-xs} s^{\frac12}\big( \kappa_+ \Pi_{+}(s) + \kappa_{-} \Pi_-(s) \big) ds \sim
 \kappa_+ \sqrt{\frac{2\pi}{n+1}} x^{\frac {2j+1-n}{2 n+2}}e^{-\frac{n+1}{n+2}x^{\frac{n+2}{n+1}}}
 \,.
 \end{equation*}
 By Theorem \ref{thm:asymptotic}, the Airy function is the fundamental solution $\Psi^{(1)}$ of the equation (\ref{eq:31ott01}).
 We notice further that the general solution of (\ref{eq:sturmDn}) can be written by means of modified Bessel functions $I_{\nu}(z)$ \cite{bateman2}
 after the transformation of the independent variable $z=\frac{s^{n+2}}{n+2}$. In fact one easily gets
 \begin{equation}
 \Phi_1(s)=s \left( c_1 I_{\frac{-1}{2+n}}(\tfrac{s^{2+n}}{2+n}) + c_2 I_{\frac1{2+n}}(\tfrac{s^{2+n}}{2+n}) \right) \, ,
 \quad c_{1}, c_2 \in \bb{C}
 \,.
 \end{equation}
 It is well known \cite{bateman2} that $I_{\nu}(z)\sim \sqrt{\frac1{2 \pi z}}e^{z} \, , \arg z <\frac{\pi}{2}$.
 Therefore, the fundamental solution $\Psi^{(1)}(x)$  written in the integral form (\ref{eq:31ott02})  depends on the sum $c_1+c_2$ only.
 Moreover, the solutions $\Psi^{(1)}_k(x)$, $|k|\leq n+1$ and $x>0$, are obtained by integrating \eqref{eq:31ott02} along the contour
 obtained rotating $\mathit{c}$ by $e^{\frac{ k \pi i }{2(n+2)}}$.
 \subsection{\texorpdfstring{$D^{(3)}_4$}{G\_2} in the standard representation}\label{sec:airy_G2}
 Let us study the $D^{(3)}_4$-Airy function in the standard representation. Namely, we want to find
 a solution for equation \eqref{eq:3nov01} in the representation $V^{(1)}=\mb C^8$ of the twisted Kac-Moody algebra
 $D_{4}^{(3)}$.
 Using the explicit form of Chevalley generators $e_i$, $i=0,1,2$, which is provided in Section \ref{sec:proof_G2}, and denoting by $\Phi_i(s)$ the components of
 $\Phi(s)$ in the standard basis of $\mb{C}^{8}$,
 equation \eqref{eq:3nov01} reads
 \begin{equation}\label{20151014:eq3}
 \left\{
 \begin{array}{ll}
 \Phi_{i+1}(s)=s\Phi_i(s)\,, & i=1,2\,,
 \\
 2\Phi_4(s)+\Phi_5(s)=2s\Phi_3(s)&
 \\
 -2\kappa\eta^2\Phi_1'(s)+\Phi_6(s)=2s\Phi_4(s)\,,&
 \\
 2 \kappa \eta \Phi_1'(s)+\Phi_6(s)=s\Phi_5(s)&
 \\
 2 \kappa \Phi_2'(s)+\Phi_7(s)=s\Phi_6&
 \\
 2 \kappa \Phi_3'(s)+\Phi_8(s)=s\Phi_7(s)&
 \\
 2 \kappa \eta\Phi_4'-\kappa\eta^2\Phi_5'=s\Phi_8(s)
 \,.&
 \end{array}
 \right.
 \end{equation}
 where $\kappa=\frac{1}{3+2\sqrt{3}}$ and $\eta=e^{i \frac{\pi}{3}}$.
 The general solution of \eqref{20151014:eq3} is
 \begin{align*}
 &\Phi_{i}(s)=s^{i-1}\Phi_1(s)\,, \qquad i=2,3\,,
 \\
 &\Phi_4(s)=\frac{s^3}{2}\Phi_1(s)-i\frac{\kappa\sqrt3}{2s}\Phi_1'(s)\,,
 \\
 &\Phi_5(s)=s^3\Phi_1(s)+i\frac{\kappa\sqrt3}s\Phi_1'(s)\,,
 \\
 &\Phi_6(s)=s^4\Phi_1(s)-\kappa\Phi_1'(s)\,,
 \\
 &\Phi_7(s)=(s^5-2\kappa)\Phi_1(s)-3\kappa s\Phi_1'(s)\,,
 \\
 &\Phi_8(s)=(s^6-6\kappa s)\Phi_1(s)-5\kappa s^2\Phi_1'(s)\,,
 \end{align*}
 where $\Phi_1(s)$ is a general solution to the equation
 \begin{equation}\label{eq:ph1d43}
 3\kappa^2(\frac1s\Phi_1'(s))'
 =(s^7-9\kappa s^2)\Phi_1(s)-6\kappa s^3\Phi_1'(s)\,.
 \end{equation}
 To analyze equation \eqref{eq:ph1d43} we can proceed as in Section \ref{sec:airy_Dn+1}.
 Letting $\Pi(s)=s^{-\frac12}e^{\frac{1}{5} (3+2 \sqrt{3}) s^5}\Phi_1(s)$ we obtain
 $$
 \Pi''(s)=\left((28+16 \sqrt{3}) s^8+ \frac{3}{4 s^2} \right)\Pi(s)\; .
 $$
 As we did in the previous subsection, we use the WKB method to study solutions to the latter equation.
 By this mean, we easily prove that in the Stokes sector $\Sigma$ defined as in Lemma \ref{lem:steepest} with $M=5$ there is a basis
 of solutions $\Pi_{\pm}(s)$ with the following asymptotic behavior
 \begin{equation*}
 \Pi_{\pm}(s) \sim s^{-2} e^{\pm \frac{2}{5} \left(2+\sqrt{3}\right) s^5}
\,.
 \end{equation*}
 Therefore the general solution of the equation (\ref{eq:ph1d43}) can be expressed in the basis $\Phi_1^{\pm}$
 with the following behavior in the Stokes sector
 \begin{align*}
 \Phi_1^+\sim s^{-\frac32} e^{ \frac{1}{5} s^5}
\,,\qquad
 \Phi_1^-\sim s^{-\frac32} e^{- \frac{7+4 \sqrt{3}}{5} s^5}
\,.
 \end{align*}
 The path of integration $\mathit{c}$ in the integral formula (\ref{eq:31ott02}) is naturally, as in the previous examples,
 the boundary of the Stokes sector $\Sigma$.
 Reasoning as in the case of the $D_{n+1}^{(2)}$-Airy function we deduce that the choice $\Phi_1^-$ leads to the zero solution of (\ref{eq:31ott01}).
 On the other hand, after the steepest descent Lemma \ref{lem:steepest} we obtain
 that the choice $\Phi_1^+$
 defines, up to a multiplicative constant, the subdominant solution $\Psi^{(1)}(x)$ of (\ref{eq:31ott01}).
 
 We conclude by noticing that the general solution of (\ref{eq:ph1d43}) can be expressed in terms of modified Bessel functions
 \begin{equation*}
 \Phi_1(s)= s \, e^{-\frac{1}{5} (3+2 \sqrt{3}) s^5} \left( c_1 \, I_{-\frac{1}{5}}\big( \frac{2}{5} (2+\sqrt{3}) s^5\big)+
 c_2 \,I_{\frac{1}{5}}\big( \frac{2}{5} (2+\sqrt{3}) s^5\big) \right) \; .
 \end{equation*}
 \begin{remark}
 Recently Bertola, Dubrovin and Yang \cite{dub15} studied equation (\ref{eq:31ott01}) and the integral formula (\ref{eq:31ott02})
 in the case of an untwisted Kac-Moody algebra $\gu$ for a fixed representation, namely  the evaluation representation at $t=1$ of
 the adjoint representation of $\mf g$.
 The motivation of their work comes from Topological Field Theory and
computations of Tau-functions  and -- quite naturally -- they called equation (\ref{eq:31ott01}) the \textit{Topological ODE}. The same equations were already
 studied in our previous work \cite{marava15} in relation with the Bethe Ansatz of quantum $\mf g$-KdV model for a simply-laced Lie algebra $\mf g$.
In \cite{dub15} the solutions of interest are the ones that grow polynomially at $\infty$ inside a fixed Stokes sector.
 Since equation (\ref{eq:31ott01}) can be reduced to the almost diagonal form (\ref{eq:quasiconstantode}), one
 deduces immediately that the sought solution
 is of the asymptotic form
 \begin{equation}\label{eq:borisODE}
 T(x) = x^{-\frac{\ad h}{h^\vee}}\left( h' + o(1) \right) \, ,
 \quad x \gg 0
 \,,
 \end{equation}
where $h \in \wh$ is the special element defined by relations (\ref{eq:hrelations}), and
$h'$ is any non-zero element of the kernel of the adjoint action of $\Lambda$, in other words an element of the Cartan subalgebra
$\lbrace g \in {\mf g} \, |\, [g,\Lambda]=0 \rbrace$. The asymptotic expansion of a topological solution is fully determined
by formula (\ref{eq:borisODE}). However, the asymptotic expansion does not determine a unique solution but an affine space
of solutions, because the asymptotic expansion does not change when adding an asymptotically small solution. In view of the discussion
after formula (\ref{eq:quasiconstantode}), this affine space is naturally associated
with $\bigoplus_{\Re \lambda_i >0} \mf{g}_{\lambda_i} $ where $\mf{g}_{\lambda_i}$ is the eigenspace of the adjoint action of $\Lambda$
with eigenvector $\lambda_i$.

Using formula (\ref{eq:borisODE}), the topological solutions of \cite{dub15} can be constructed from the subdominant solutions $\Psi^{(i)}$'s
considered in \cite{marava15}. For example, choosing $\mf g$ of type $A_n$ and considering the untwisted algebra $A_n^{(1)}$, then due
to the decomposition of $\mf g$-modules
 $ L(\omega_1) \otimes L(\omega_n) \cong \mf g \oplus  \bb{C} $, a full set of asymptotic expansions (\ref{eq:borisODE}) can be written
in terms of $\Psi_k^{(1)}(x)\otimes \Psi_{k+\e_k \frac{n+1}2}^{(n)}(x)$, $k= \lfloor\frac{-n}{2}\rfloor , \dots, \lfloor\frac{n}{2}\rfloor$; here
$\e_k=-1$ if $k>0$ and $\e_k=1$ if $k \leq 0$.
It follows that one can write the Topological Tau-functions by means
of the Bethe Ansatz solutions $Q^{(1)},\dots, Q^{(n)}$. It would be interesting
 to understand whether this apparent coincidence is the manifestation of a relation between the underlying physical theories, also with the aim of
 fixing the ambiguity in the choice of topological solutions.
 We will tackle this question in
 a forthcoming paper.
 \end{remark}
 

\def\cprime{$'$} \def\cprime{$'$} \def\cprime{$'$} \def\cprime{$'$}
  \def\cprime{$'$} \def\cprime{$'$} \def\cprime{$'$} \def\cprime{$'$}
  \def\cprime{$'$} \def\cprime{$'$} \def\cydot{\leavevmode\raise.4ex\hbox{.}}
  \def\cprime{$'$} \def\cprime{$'$} \def\cprime{$'$}

\end{document}